\documentclass[11pt]{article}
\usepackage[utf8]{inputenc}

\usepackage{amsmath,amssymb,amsthm}
\usepackage{mathtools}
\usepackage{fullpage}
\usepackage{color}
\usepackage{graphicx}
\usepackage[colorlinks=true, citecolor=blue]{hyperref}
\usepackage[font=small,labelfont=bf,width=\textwidth]{caption}
\usepackage{subcaption}
\usepackage[usenames,dvipsnames]{xcolor}
\usepackage{tikz}

\usepackage{fullpage}

\RequirePackage[sc]{mathpazo}
\RequirePackage[T1]{fontenc}

\usepackage[capitalize,nameinlink,noabbrev]{cleveref}

\newtheorem{theorem}{Theorem}
\newtheorem{lemma}[theorem]{Lemma}

\usepackage[maxbibnames=99,sorting=nyt,giveninits=true]{biblatex}
\renewbibmacro{in:}{. In:}
\DeclareFieldFormat[misc]{title}{\mkbibquote{#1}}
\renewcommand*{\bibfont}{\small}
\DeclareFieldFormat[article]{volume}{\mkbibbold{#1}}
\DeclareFieldFormat{url}{\textsc{url}: \href{#1}{#1}}
\DeclareFieldFormat{doi}{\textsc{doi}: \href{http://dx.doi.org/#1}{#1}}
\DeclareFieldFormat{eprint:arxiv}{arXiv:\href{https://arxiv.org/abs/#1}{#1}}

\theoremstyle{remark}
\newtheorem*{note}{Note}

\newcommand{\ols}{\overline{s}}
\DeclareMathOperator{\bigO}{O}
\DeclareMathOperator{\littleo}{o}
\newcommand{\hs}{\hat{s}}
\newcommand{\hz}{\hat{z}}
\newcommand{\Uw}{U^{[w]}_w}
\newcommand{\Dzero}{D^{[0]}_w}

\usepackage{authblk}

\title{Semi-flexible directed polymers in a strip with attractive walls}

\author{Nicholas R. Beaton\thanks{\href{mailto:nrbeaton@unimelb.edu.au}{nrbeaton@unimelb.edu.au}}}
\author{Leo Li\thanks{\href{mailto:leo.li2@unimelb.edu.au}{leo.li2@unimelb.edu.au}}}
\affil{School of Mathematics and Statistics, The University of Melbourne, Australia}

\author{Jonathon Liu\thanks{\href{mailto:leol3@student.unimelb.edu.au}{leol3@student.unimelb.edu.au}}}
\affil{School of Mathematics and Statistics, The University of Sydney, Australia}

\author{Thomas Wong\thanks{\href{mailto:thomas.wong@hw.ac.uk}{thomas.wong@hw.ac.uk}}}
\affil{Department of Mathematics, Heriot Watt University, Edinburgh, Scotland}

\begin{document}

\maketitle

\begin{abstract}
    We study a model of a semiflexible long chain polymer confined to a two-dimensional slit of width $w$, and interacting with the walls of the slit. The interactions with the walls are controlled by Boltzmann weights $a$ and $b$, and the flexibility of the polymer is controlled by another Boltzmann weight $c$. This is a simple model of the steric stabilisation of colloidal dispersions by polymers in solution. We solve the model exactly and compute various quantities in $(a,b,c)$-space, including the free energy and the force exerted by the polymer on the walls of the slit. In some cases these quantities can be computed exactly for all $w$, while for others only asymptotic expressions can be found. Of particular interest is the zero-force surface -- the manifold in $(a,b,c)$-space where the free energy is independent of $w$, and the loss of entropy due to confinement in the slit is exactly balanced by the energy gained from interactions with the walls.
\end{abstract}

\section{Introduction}\label{sec:intro}

Polymers in dilute solution, confined to a narrow channel or between two plates, lose configurational entropy and thus exert a repulsive (outward) force on the walls of the confined space. However, when the polymers experience an attractive interaction with the walls, there is also a force in the reverse direction, and the polymers can work to pull the walls together. This is seen in the process of steric stablisation, where polymer molecules in solution with (much larger) colloidal particles are attracted to the particles but then serve to hold them apart and maintain the stability of the solution~\cite{daoud_statistics_1977,nagele_1997,napper_1983,rudhardt_direct_1998,verma_entropic_1998,zhulina_theory_1990}.

Self-avoiding walks (SAWs) are a classical model of long-chain polymers in dilute solution. To model a polymer in confinement one can restrict a SAW (in $\mathbb{Z}^d$, say) to a strip (in two dimensions) or slab (in three dimensions) of width $w$. Some rigorous results are known about this model  (see \cite{janse_van_rensburg_self-avoiding_2006} for a thorough treatment). Define the $w$-slab 
\begin{equation}
    \mathbb{S}_w = \{(x_1,\dots,x_d)\in\mathbb{Z}^d \,|\, 0\leq x_d\leq w\},
\end{equation}
and let $\mathcal{W}_{w,n}$ be the set of $n$-step SAWs which start at the origin and stay in $\mathbb{S}_w$. For $\phi\in\mathcal{W}_{w,n}$, let $m_a(\phi)$ (resp.~$m_b(\phi)$) be the number of vertices of $\phi$ in the hyperplane $x_d=0$ (resp.~$x_d=w$), excluding the origin. Then the size-$n$ partition function of the model is
\begin{equation}
    C_{w,n}(a,b) = \sum_{\phi\in\mathcal{W}_{w,n}} a^{m_a(\phi)} b^{m_b(\phi)}.
\end{equation}

It is known~\cite{janse_van_rensburg_self-avoiding_2006} that the limiting free energy
\begin{equation}
    \kappa_w(a,b) = \lim_{n\to\infty}\frac1n\log C_{w,n}(a,b)
\end{equation}
exists and is a convex, continuous and almost-everywhere differentiable function of $\log a$ and $\log b$. When $b=1$ (ie.~there are interactions with only the bottom wall) it has been shown that $\kappa_w(a,1) \to \kappa(a)$ as $w\to\infty$, where $\kappa(a)$ is the free energy of adsorbing SAWs in a half-space. Moreover it was conjectured that $\kappa_w(a,b) \to \max\{\kappa(a),\kappa(b)\}$ as $w\to\infty$.

For finite $w$ it was also conjectured that there is a zero-force curve in the $a$-$b$ plane where $\kappa_w(a,b)=\kappa_{w-1}(a,b)$. Below this curve $\kappa_w(a,b)>\kappa_{w-1}(a,b)$ (corresponding to a repulsive force between the planes) and above it $\kappa_w(a,b)<\kappa_{w-1}(a,b)$ (corresponding to an attractive force). As $w\to\infty$ these curves were predicted to approach a limiting curve, with asymptotes $a=1$ and $b=1$ and passing through $(a,b)=(a_\mathrm{c},a_\mathrm{c})$, where $a_\mathrm{c}$ is the critical point for adsorbing SAWs in a half-space.

A Monte Carlo study of this model was conducted in~\cite{rensburg_self-avoiding_2005}. Similar results to~\cite{janse_van_rensburg_self-avoiding_2006} were found. The authors also approximated the force between the plates as
\begin{align}
    F_{w,n}(a,b) &= \frac1n\log C_{w+1,n}(a,b) - \frac1n\log C_{w,n}(a,b) \\
    &= \kappa_{w+1,n}(a,b) - \kappa_{w,n}(a,b),
\end{align}
and, with $F_w(a,b) = \lim_{n\to\infty} F_{w,n}(a,b)$, confirmed a prediction of Daoud and de Gennes~\cite{daoud_statistics_1977} that 
\begin{equation}\label{eqn:daoud_degennes}
    F_{w}(1,1) \sim \mathrm{const.}\times w^{-1-1/\nu}
\end{equation}
where $\nu$ is the metric exponent for dilute SAWs (expected to be $\frac34$ in two dimensions and $\approx 0.588$ in three dimensions). This scaling form is also expected to hold for $a,b<a_\mathrm{c}$.

In~\cite{brak_directed_2005} a two-dimensional directed version of this model was considered (see also~\cite{wong_enumeration_2015}). The setup is the same as described above, but now the walks start at the origin and may only take steps $(1,1)$ or $(1,-1)$. It was found that as $w\to\infty$,
\begin{equation}
    \kappa_w(a,b) \to \max\{\kappa(a),\kappa(b)\} = \begin{cases} \log 2 & a,b\leq 2 \\
    \log \left(\frac{a}{\sqrt{a-1}}\right) & a > \max\{2,b\} \\
    \log \left(\frac{b}{\sqrt{b-1}}\right) & \text{otherwise.}\end{cases}
\end{equation}
The force on the walls was defined as $F_w(a,b) = \frac{\partial}{\partial w}\kappa_w(a,b)$, and it was found that $F_w(a,b) = 0$ along the curve $ab-a-b=0$. Above this curve the force is positive and short-ranged (decays exponentially with $w$), while below the curve the force is negative and can be long- or short-ranged (decays polynomially with $w$), depending on whether one of $a$ or $b$ is greater than $a_\mathrm{c} = 2$. For $a,b<2$, the asymptotic form~\eqref{eqn:daoud_degennes} also holds for directed walks (the corresponding value of $\nu$ is $\frac12$).

In this paper we generalise the work of~\cite{brak_directed_2005} by taking the \emph{flexibility} of the polymers into account. Real-world polymers can have a certain level of stiffness or rigidity, and this affects phase transitions and critical behaviour. A standard method for taking this into account in mathematical polymer models is to assign weights to consecutive segments of the polymer according to the angle between the segments -- for example, one can assign a Boltzmann weight $c$ to consecutive pairs of collinear segments (this is the method used here, and see for example~\cite{hsu_semi-flexible_2013,krawczyk_semi-flexible_2015,owczarek_exact_2009,zivic_semiflexible_2018}). As the weight is increased, the average walk tends to have more long straight segments and fewer bends.

The layout of the paper is as follows. In \cref{sec:the_model} we define the model and some key quantities, and work through two different methods for solving it exactly. In \cref{sec:symmetric} we consider the special case $a=b$, analysing the asymptotic behaviour in different parts of $(a,c)$-space. In \cref{sec:asymmetric} we study the full model in $(a,b,c)$-space. In \cref{sec:sampling} we use the method of generating trees to randomly sample long walks for a variety of different $(a,b,c)$ values. Some concluding remarks are given in \cref{sec:conclusion}.

\section{The model}\label{sec:the_model}

\subsection{Definitions}\label{ssec:definitions}

We consider a directed walk propagating along a slit. Specifically, consider a walk beginning at the origin, which takes steps $(1, \pm 1)$. Further, fix a $w\in \mathbb{N}$ and restrict the allowed vertices to
$\{(x,y)\in \mathbb{Z}^2 \ | \ 0\leq y \leq w \}.$
Call $w$ the width of the slit, and let $\mathcal{W}_w$ be the set of all walks in the slit of height $w$.

For any such walk $\phi$, let $|\phi|$ denote the length of this walk, i.e. the number of steps, which is also the $x$-coordinate of the terminating vertex. Physical interactions are incorporated into the model by associating an energy to each walk, with energy contributions (Boltzmann weights) arising from three kinds of interactions. Walks gain weight $\epsilon_a$ for each contact with the bottom wall (excluding the initial contact at the origin), weight $\epsilon_b$ for each contact with the top wall, and weight $\epsilon_c$ for each pair of consecutive up or down steps (`stiffness points'). If $\phi$ touches the bottom wall $m_a(\phi)$ times, the top wall $m_b(\phi)$ times, and has $m_c(\phi)$ stiffness points, then its associated energy is $m_a(\phi)\epsilon_a+m_b(\phi)\epsilon_b+m_c(\phi)\epsilon_c.$ See \cref{fig:example}.

\begin{figure}
    \centering
    \resizebox{\textwidth}{!}{
    \begin{tikzpicture}
    \tikzset{bc/.style = {circle, draw, line width=4pt, draw=red, fill=blue, inner sep=3pt}};
    \tikzset{tc/.style = {circle, draw, line width=4pt, draw=green, fill=blue, inner sep=3pt}};
    \tikzset{sp/.style = {draw, line width=3pt, draw=orange, fill=blue, inner sep=4pt}};
    \draw [gray] (0,0) grid (25,6);
    \draw [gray, fill=gray] (0,0) rectangle (25,-0.2);
    \draw [gray, fill=gray] (0,6) rectangle (25,6.2);
    \draw [blue, line width=0.1cm] (0,0) -- (1,1) -- (2,2) -- (3,1) -- (4,0) -- (5,1) -- (6,0) -- (7,1) -- (8,2) -- (9,3) -- (10,4) -- (11,3) -- (12,4) -- (13,5) -- (14,6) -- (15,5) -- (16,6) -- (17,5) -- (18,4) -- (19,3) -- (20,4) -- (21,3) -- (22,2) -- (23,1) -- (24,0) -- (25,1);
    \node [bc] at (4,0) {};
    \node [bc] at (6,0) {};
    \node [bc] at (24,0) {};
    \node [tc] at (14,6) {};
    \node [tc] at (16,6) {};
    \node [sp] at (1,1) {};
    \node [sp] at (3,1) {};
    \node [sp] at (7,1) {};
    \node [sp] at (8,2) {};
    \node [sp] at (9,3) {};
    \node [sp] at (12,4) {};
    \node [sp] at (13,5) {};
    \node [sp] at (17,5) {};
    \node [sp] at (18,4) {};
    \node [sp] at (21,3) {};
    \node [sp] at (22,2) {};
    \node [sp] at (23,1) {};
    \end{tikzpicture}
    }
    \caption{A semiflexible directed path in a strip of width 6. This path has length 25, three contacts with the bottom wall (red circles), two contacts with the top wall (green circles), and 12 stiffness points (orange squares). It thus contributes weight $a^3b^2c^{12}$ to $Z_{6,25}(a,b,c)$.}
    \label{fig:example}
\end{figure}
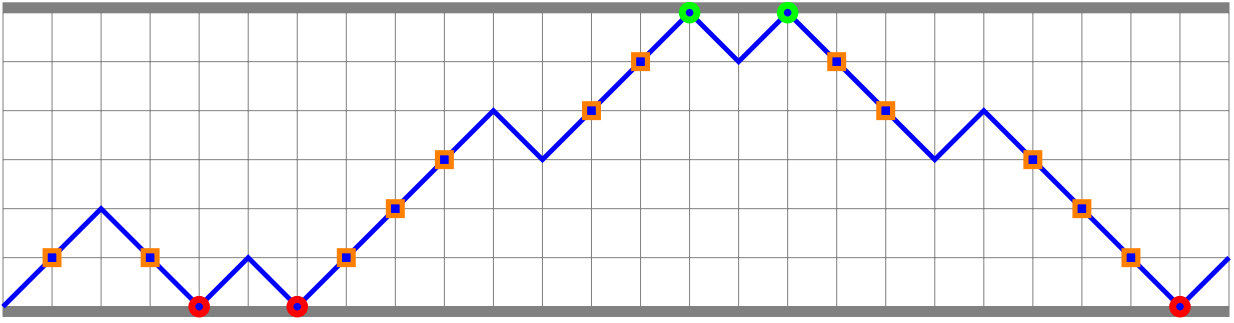

The canonical partition function for this system will be 
\begin{equation}
    Z_{w,n}(a,b,c)=\sum_{\phi\in\mathcal{W}_{w,n}}\exp\left(-\frac{m_a(\phi)\epsilon_a+m_b(\phi)\epsilon_b+m_c(\phi)\epsilon_c}{k_B T} \right)=
    \sum_{\phi\in\mathcal{W}_{w,n}} a^{m_a(\phi)}b^{m_b(\phi)}c^{m_c(\phi)}
\end{equation}
where $\mathcal{W}_{w,n}$ is the set of all length $n$ walks in the width $w$ strip, $k_B$ is the Boltzmann constant, $T$ is the absolute temperature, and $a=\exp\left(-\frac{\epsilon_a}{k_B T}\right)$, $b=\exp\left(-\frac{\epsilon_b}{k_B T}\right)$, and
$c=\exp\left(-\frac{\epsilon_c}{k_B T}\right)$ are Boltzmann weights.

Throughout the paper we always assume that $a,b,c>0$.

\begin{lemma}\label{lem:free_energy}
The \emph{free energy}
\begin{equation}
\kappa_w(a,b,c)=\lim_{n\to \infty} \frac{1}{n}\log Z_{w,n}(a,b,c)
\end{equation}
exists for all $a,b,c>0$. It is a continuous, almost-everywhere differentiable, and strictly increasing function of $a,b,c$.
\end{lemma}

To prove \cref{lem:free_energy} it will be useful to introduce another statistic on walks. Let $h(\phi)$ be the $y$-coordinate of the terminating vertex of $\phi$, i.e.~the final height of $\phi$. The corresponding partition function is then
\begin{equation}
    Z_{w,n,k}(a,b,c) = \sum_{\substack{\phi\in\mathcal{W}_{w,n} \\ h(\phi) = k}} a^{m_a(\phi)}b^{m_b(\phi)}c^{m_c(\phi)}
\end{equation}
Note that $Z_{w,n,k}(a,b,c) = 0$ if $k\not\equiv n \,(\text{mod } 2)$. 

We will sometimes refer to walks which end on the bottom wall as \emph{loops}, and walks which end on the top wall as \emph{bridges}. Loops and bridges have partition functions $Z_{w,n,0}(a,b,c)$ and $Z_{w,n,w}(a,b,c)$ respectively.

\begin{proof}[Proof of \cref{lem:free_energy}]
We begin by noting that if $\phi_1$ and $\phi_2$ are two walks with $h(\phi_1) = h(\phi_2) = 0$, then they can be concatenated to form a longer walk $\phi$, where
\begin{align}
  & &  |\phi| &= |\phi_1| + |\phi_2|, & & \\
  & &  \qquad h(\phi) &= 0, & & \\
    m_a(\phi) &= m_a(\phi_1) + m_a(\phi_2), & m_b(\phi) &= m_b(\phi_1) + m_b(\phi_2), & m_c(\phi) &= m_c(\phi_1) + m_c(\phi_2).
\end{align}
It follows that, for even $m$ and $n$,
\begin{equation}
    Z_{w,m+n,0}(a,b,c) \geq Z_{w,m,0}(a,b,c)Z_{w,n,0}(a,b,c),
\end{equation}
and hence $-\log Z_{w,n,0}(a,b,c)$ is a subadditive sequence in $n$. A standard result on subadditive sequences~\cite{hille_functional_1948} then implies that the limit
\begin{equation}\label{eqn:fe_for_height0}
    \kappa_{w,0}(a,b,c) = \lim_{n\to\infty} \frac1n \log Z_{w,n,0}(a,b,c)
\end{equation}
exists, where the limit is taken through even $n$. The fact that $\kappa_{w,0}(a,b,c)$ is continuous and almost-everywhere differentiable can also be proved using standard techniques -- see for example~\cite{janse_van_rensburg_self-avoiding_2006}. The fact that $\kappa_{w,0}(a,b,c)$ is strictly increasing follows from the fact that it is the spectral radius of a finite irreducible matrix (see e.g.~\cite[Chapter 8]{meyer2000}).

Now let $(h_n)_{n\geq0}$ be a sequence with $h_n \in [0,w]$ and $h_n \equiv n\,(\text{mod }2)$. A walk of length $n-h_n$ ending at height 0 can be extended by $h_n$ steps to become a walk of length $n$ ending at height $h_n$, with the addition of at most $h_n$ stiffness sites and at most one top contact. If we set $b_+ = \max\{1,b^{-1}\}$ and $c_+=\max\{1,c^{-1}\}$, then this implies
\begin{equation}\label{eqn:Znh_lowerbound}
    Z_{w,n-h_n,0}(a,b,c) \leq b_+ c_+^{h_n} Z_{w,n,h_n}(a,b,c).
\end{equation}
Similarly, a walk of length $n$ ending at height $h_n$ can be extended by $h_n$ steps to become a walk ending at height 0, again with the addition of at most $h_n$ stiffness sites and at most one bottom contact. With $a_+ = \max\{1,a^{-1}\}$, we get
\begin{equation}\label{eqn:Znh_upperbound}
    Z_{w,n,h_n}(a,b,c) \leq a_+ c_+^{h_n} Z_{w,n+h_n,0}(a,b,c).
\end{equation}
For each of~\eqref{eqn:Znh_lowerbound} and~\eqref{eqn:Znh_upperbound} take the logs, divide by $n$, and take the $\liminf$ and $\limsup$ respectively. By~\eqref{eqn:fe_for_height0}, it follows that
\begin{equation}\label{eqn:Znh_fe}
    \lim_{n\to\infty} \frac1n \log Z_{w,n,h_n}(a,b,c) = \kappa_{w,0}(a,b,c).
\end{equation}

Finally let $h^-_n \equiv h^-_n(a,b,c)$ be such that $Z_{w,n,h^-_n}(a,b,c) \leq Z_{w,n,h}(a,b,c)$ for all $0\leq h \leq w$, and likewise let $h^+_n \equiv h^+_n(a,b,c)$ be such that $Z_{w,n,h^+_n}(a,b,c) \geq Z_{w,n,h}(a,b,c)$ for all $0\leq h \leq w$ (all constrained so that $h^-_n \equiv h^+_n \equiv h \equiv n \,(\text{mod } 2)$). Then
\begin{align}
    \left\lceil\frac{w+1}{2}\right\rceil Z_{w,n,h^-_n}(a,b,c) \leq Z_{w,n}(a,b,c) \leq \left\lceil\frac{w+1}{2}\right\rceil Z_{w,n,h^+_n}(a,b,c).
\end{align}
Again take logs, divide by $n$ and take the limit. By~\eqref{eqn:Znh_fe}, the result follows, with $\kappa_{w,0}(a,b,c) = \kappa_w(a,b,c)$.
\end{proof}

From the free energy, we can obtain the effective force exerted on the walls of the slit due to the polymer,
\begin{equation}
    \mathcal{F}_w(a,b,c) = \frac{\partial}{\partial w}\kappa_w(a,b,c).
\end{equation}
In particular, we are interested in `zero-force' curves and surfaces -- the loci of points $(a,b,c)$ where $\mathcal{F}_w(a,b,c)=0$.

As per work in previous papers~\cite{brak_directed_2005,wong_enumeration_2015}, we will make use of the generating function of the system,
\begin{equation}
    G_w(a,b,c;z) \equiv G_w = \sum_{n=0}^\infty Z_{w,n}(a,b,c)z^n.
\end{equation}
Viewed as a power series in $z$, this generating function has a nonzero radius of convergence $R_w(a,b,c)$ about the origin, and has a dominant singularity at $z=z_w(a,b,c)$ on the positive real axis. There is a relation between the dominant singularity $z_w(a,b,c)$ and the free energy,
\begin{equation}\label{eqn:fe(z)}
    \kappa_w(a,b,c)=-\log z_w(a,b,c).
\end{equation}

We will also use a generalisation of $G_w$ which takes into account the final height of walks:
\begin{equation}
    F_w(a,b,c;z,s) \equiv F_w(s) = \sum_{n=0}^\infty \sum_{h=0}^w Z_{w,n,h}(a,b,c)z^n s^h.
\end{equation}
Of course $F_w(1) = G_w$.

Partition $\mathcal{W}_w$ into `down walks' $\mathcal{D}_w$ and `up walks' $\mathcal{U}_w$ -- the former being the set of walks with final step $(1,-1)$, and the latter being the set of walks with final step $(1,1)$. It is consistent to place the 0-length walk in $\mathcal{D}_w$ since all walks ending at $y=0$ are $\mathcal{D}_w$ walks.

Now let $D_w(a,b,c;z,s) \equiv D_w(s)$ be the generating function for down walks, and likewise let $U_w(a,b,c;z,s) \equiv U_w(s)$ be the generating function for up walks. In particular, $F_w(s)=D_w(s)+U_w(s)$. These two generating functions are needed to construct a recurrence relation.

For a formal power series $f(z)=\sum_{n=0}^\infty a_n z^n$, we use the notation $[z^k]f(z)=a_k$. Then $[s^h]F_w(s)$ is the generating function for walks ending at some height $h$.

\begin{lemma}\label{lem:fe_all_gfs}
For $0\leq h \leq w$, the generating functions 
\begin{equation}
    F_w(1), D_w(1), U_w(1), [s^h]F_w(s), [s^h]D_w(s) \text{ and } [s^h]U_w(s)
\end{equation}
all have the same radius of convergence, namely $z_w(a,b,c) = \exp(-\kappa_w(a,b,c))$. (Excluding the trivial cases $[s^w]D_w(s) = [s^0]U_w(s) = 0$.)
\end{lemma}

\begin{proof}[Proof (sketch)]
The proof of \cref{lem:free_energy} already established this for $F_w(1)$. For $D_w(1)$ and $U_w(1)$, the proof works in much the same way -- start with walks ending at height 0, append steps to get up to height $h_n$ (ending with $(1,1)$ or $(1,-1)$ as required), and then again to get back to height 0.

For the $[s^h]$ generating functions, we must take some care with parity issues. Assume for now that $h$ is even. Then $[s^h]F_w(s)$ contains only even powers of $z$. For even $n$, by the same arguments used in \cref{lem:free_energy}, we have
\begin{equation}
    Z_{w,n-h,0}(a,b,c) \leq b_+c_+^h Z_{w,n,h}(a,b,c) \leq a_+b_+c_+^{2h}Z_{w,n+h,0}(a,b,c).
\end{equation}
Take logs, divide by $n$ and take the limit through even values of $n$ only. By~\eqref{eqn:fe_for_height0} the limit is $\kappa_{w,0}(a,b,c) = \kappa_w(a,b,c)$, and so the result follows for $[s^h]F_w(s)$ with even $h$. 

For odd $h$, and for $[s^h]D_w(s)$ and $[s^h]U_w(s)$, the proof is similar.
\end{proof}

\subsection{Comparison with the half-plane model}

Let $\mathcal{W}^+_n$ be the set of directed half-plane walks of length $n$ which start on the surface, and for such a walk $\phi$ let $m_a(\phi)$ be the number of visits to the surface (excluding the initial vertex) and $m_c(\phi)$ be the number of stiffness sites. We then have the partition function
\begin{equation}
    Z^+_n(a,c) = \sum_{\phi\in\mathcal{W}^+_n} a^{m_a(\phi)}c^{m_c(\phi)}
\end{equation}
and generating function
\begin{equation}
    G^+(a,c;z) = \sum_{n=0}^\infty Z^+_n(a,c)z^n.
\end{equation}
It will be useful to also define corresponding quantities for half-plane loops:
\begin{equation}
    Z^+_{n,0}(a,c) = \sum_{\substack{\phi\in\mathcal{W}^+_n \\ h(\phi)=0}} a^{m_a(\phi)}c^{m_c(\phi)} \qquad\text{and}\qquad G^+_0(a,c;z) = \sum_{n=0}^\infty Z^+_{n,0}(a,c)z^n.
\end{equation}

It is straightforward to show that
\begin{align}
    G^+_0(a,c;z) &= \frac{2-a-az^2+ac^2z^2-a\sqrt{(1-z^2+c^2z^2)^2-4c^2z^2}}{2(1-a-az^2+a^2z^2+ac^2z^2} \\
    G^+(a,c;z) &= \frac{(1-a-acz)(1-2cz-z^2+c^2z^2)+(1-a+acz)\sqrt{(1-z^2+c^2z^2)^2-4c^2z^2}}{2(1-z-cz)(1-a-az^2+a^2z^2+ac^2z^2)}.
\end{align}

Both $G^+_0$ and $G^+$ have the same dominant singularity,
\begin{equation}\label{eqn:halfplane_domsing}
    z^+(a,c) = \begin{cases} \frac{1}{c+1} & a \leq c+1 \\ \frac{\sqrt{a-1}}{\sqrt{a(a+c^2-1)}} & a>c+1, \end{cases}
\end{equation}
and we define the corresponding free energy $\kappa^+(a,c) = -\log z^+(a,c)$.

The following theorem shows that when the walks interact with only one side of the width-$w$ strip, as $w\to\infty$ we simply obtain the half-plane model.
\begin{theorem}\label{thm:strip_to_halfplane}
For all $a,c>0$,
\begin{equation}\label{eqn:kappaw_limit}
    \lim_{w\to\infty} \kappa_w(a,1,c) = \kappa^+(a,c).
\end{equation}
\end{theorem}
\begin{proof}
By \cref{lem:fe_all_gfs} and~\eqref{eqn:halfplane_domsing} we can just focus on loops. In the half-plane and $w$-strip we have respectively
\begin{equation}
    \kappa^+(a,c) = \lim_{n\to\infty} \frac1n\log Z^+_{n,0}(a,c) \qquad\text{and}\qquad \kappa_w(a,1,c) = \lim_{n\to\infty} \frac1n\log Z_{w,n,0}(a,1,c).
\end{equation}
Since $Z_{w,n,0}(a,1,c) \leq Z^+_{n,0}(a,c)$ and $Z_{w,n,0}(a,1,c) \leq Z_{w+1,n,0}(a,1,c)$ (walks in a $w$-strip with no interactions on the top wall are also walks in a half-plane and in a $(w+1)$-strip), the limit in~\eqref{eqn:kappaw_limit} exists and is at most $\kappa^+(a,c)$.

To show that this limit is equal to $\kappa^+(a,c)$, first note that since $Z^+_{n,0}(a,c)$ and $Z_{w,n,0}(a,1,c)$ are supermultiplicative sequences (loops can be concatenated), we have
\begin{equation}\label{eqn:loops_supermult_sup}
    \kappa^+(a,c) = \sup_{n} \left\{\frac1n\log Z^+_{n,0}(a,c)\right\} \qquad\text{and}\qquad \kappa_w(a,1,c) = \sup_{n} \left\{\frac1n\log Z_{w,n,0}(a,1,c)\right\}.
\end{equation}
By definition of the limit, for any $\epsilon>0$ there exists $N$ such that
\begin{equation}
    0 \leq \kappa^+(a,c) - \frac1n\log Z^+_{n,0}(a,c) < \epsilon \qquad\text{for all } n\geq N.
\end{equation}
Choose an $\epsilon>0$ and take $N$ as above, and observe that $\frac1N\log Z_{w,N,0}(a,1,c) = \frac1N\log Z^+_{N,0}(a,c)$ for $w\geq \frac{N}{2}$. Hence
\begin{equation}
    0 \leq \kappa^+(a,c) - \frac1N\log Z_{w,N,0}(a,1,c) < \epsilon \qquad\text{for } w\geq\frac{N}{2}.
\end{equation}
But now by~\eqref{eqn:loops_supermult_sup}, $\frac1N\log Z_{w,N,0}(a,1,c)$ is a lower bound for $\kappa_w(a,1,c)$, so in fact
\begin{equation}
    0 \leq \kappa^+(a,c) - \kappa_w(a,1,c) < \epsilon.
\end{equation}
The result follows.
\end{proof}

We state without proof another result regarding half-plane walks, which can be easily derived by incorporating a variable into the above generating functions which tracks the endpoint height of walks. 

\begin{lemma}\label{lem:endpoint_height}
    For $0 < a < c+1$ consider the Boltzmann distribution on half-plane walks of length $n$: each walk $\phi\in\mathcal{W}^+_n$ is sampled with probability 
    \begin{equation}
        \mathbb{P}_n(\phi) = \frac{a^{m_a(\phi)}c^{m_c(\phi)}}{Z_n^+(a,c)}.
    \end{equation}
    Let $\langle \cdot \rangle_n$ denote expectation with respect to this distribution, and let $h(\phi)$ be the endpoint height of a walk $\phi$ as per \cref{ssec:definitions}. Then
    \begin{equation}
        \langle h \rangle_n = \sqrt{\frac{c\pi n}{2}}(1+\littleo(1)).
    \end{equation}
\end{lemma}

\subsection{Solving the generating functions}\label{ssec:solving_gfs}

To find the dominant singularity and therefore obtain the free energy, we will construct a pair of functional equations satisfied by the generating functions $U_w(s)$ and $D_w(s)$, and solve these equations to obtain the explicit expression for one of them. It does not matter which generating function we solve, since by \cref{lem:fe_all_gfs} they all have the same dominant singularity. 

For brevity, we use the notation $D^{[h]}_w = [s^h]D_w(s)$ and $U^{[h]}_w = [s^h]U_w(s)$.

\begin{theorem}\label{thm:main_functional_equations}
The generating functions $U_w$ and $D_w$ satisfy the functional equations
\begin{align} 
    D_w(s) &=  1+z\ols\Big(cD_w(s)+U_w(s)\Big)+z(a-1)\Big(cD^{[1]}_w + U^{[1]}_w\Big)-z\ols c \Dzero, \label{eq:2} \\
    U_w(s) &= zs\Big(D_w(s)+cU_w(s)\Big)+zs^w(b-1)\Big(D^{[w-1]}_w+cU^{[w-1]}_w\Big)-zs^{w+1} c\Uw. \label{eq:1}
\end{align}
\end{theorem}

\begin{proof}
Consider the first equation. Roughly speaking, every $\mathcal{U}$ walk is either another $\mathcal{U}$ walk with an up step appended to the end, or a $\mathcal{D}$ walk with an up step appended to the end. This leads to the terms $+zsD_w(s)$ and $+zscU_w(s)$, since appending an up step to a $\mathcal{U}$ walk leads to an extra stiffness point. See \cref{fig:construction}.

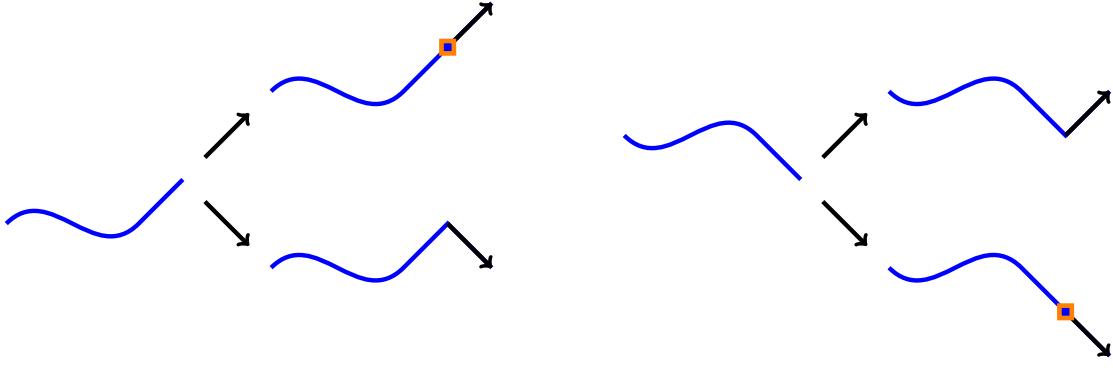
\begin{figure}
    \centering
    \resizebox{0.9\textwidth}{!}{
    \begin{tikzpicture}
    \tikzset{bc/.style = {circle, draw, line width=4pt, draw=red, fill=blue, inner sep=3pt}};
    \tikzset{tc/.style = {circle, draw, line width=4pt, draw=green, fill=blue, inner sep=3pt}};
    \tikzset{sp/.style = {draw, line width=3pt, draw=orange, fill=blue, inner sep=4pt}};
    \draw [blue, line width=0.1cm] (1,2) .. controls (2,3) and (3,1) .. (4,2) -- (5,3);
    
    \draw [black, line width=0.1cm, ->] (5.5,3.5) -- (6.5,4.5);
    \draw [blue, line width=0.1cm] (7,5) .. controls (8,6) and (9,4) .. (10,5) -- (11,6) -- (12,7);
    \draw [black, line width=0.1cm, ->] (11,6) -- (12,7); %
    \node [sp] at (11,6) {};
    
    \draw [black, line width=0.1cm, ->] (5.5,2.5) -- (6.5,1.5);
    \draw [blue, line width=0.1cm] (7,1) .. controls (8,2) and (9,0) .. (10,1) -- (11,2) -- (12,1);
    \draw [black, line width=0.1cm, ->] (11,2) -- (12,1); %
    
    \draw [blue, line width=0.1cm] (15,4) .. controls (16,3) and (17,5) .. (18,4) -- (19,3);
    
    \draw [black, line width=0.1cm, ->] (19.5,3.5) -- (20.5,4.5);
    \draw [blue, line width=0.1cm] (21,5) .. controls (22,4) and (23,6) .. (24,5) -- (25,4) -- (26,5);
    \draw [black, line width=0.1cm, ->] (25,4) -- (26,5); %
    
    \draw [black, line width=0.1cm, ->] (19.5,2.5) -- (20.5,1.5);
    \draw [blue, line width=0.1cm] (21,1) .. controls (22,0) and (23,2) .. (24,1) -- (25,0) -- (26,-1);
    \draw [black, line width=0.1cm, ->] (25,0) -- (26,-1); %
    \node [sp] at (25,0) {};
    
    \end{tikzpicture}
    }
    \caption{Illustrating part of \cref{thm:main_functional_equations}. In general an up step or a down step can be appended to an existing walk, but one must take into account the previous step in order to determine if a $c$ weight is also accrued.}
    \label{fig:construction}
    \end{figure}

Issues at the top of the strip lead to boundary terms. Firstly, if an up step is appended to a walk terminating on $y=w$, the new walk will leave the slit. This contribution must be cancelled out by $-zs^{w+1}c\Uw$. Secondly, appending an up step to a walk ending on $y=w-1$ will lead to an extra $b$ interaction. So $zs^wD^{[w-1]}_w$ must be replaced with $zs^wbD^{[w-1]}_w$, and likewise $zs^wc\Uw$ with $zs^wbcU^{[w-1]}_w$. The equation for $D_w(s)$ is constructed analogously -- the only difference is that the 0-length walk must be added as $+1$, since it is not constructed by appending a step to an existing walk.
\end{proof}

Additional relations can be obtained by taking the $s^w$ coefficient of~\eqref{eq:1} and the $s^0$ coefficient of~\eqref{eq:2}:
\begin{align}
    \Dzero &= 1+z\Big(cD^{[1]}_w+U^{[1]}_w\Big)+z(a-1)\Big(cD^{[1]}_w+U^{[1]}_w\Big),\\
    \Uw &=  z\Big(D^{[w-1]}_w+cU^{[w-1]}_w\Big)+z(b-1)\Big(D^{[w-1]}_w+cU^{[w-1]}_w\Big).
\end{align}
These allow us to eliminate $D^{[w-1]}_w+cU^{[w-1]}_w$ and $cD^{[1]}_w+U^{[1]}_w$ from the original equations, to get
\begin{align}
    D_w(s) &= \frac{1}{a}+z\ols\Big(cD_w(s)+U_w(s)\Big)+\left(1-\frac{1}{a}-cz\ols\right)\Dzero,\\
    U_w(s) &= zs\Big(D_w(s)+cU_w(s)\Big)+s^w\left(1-\frac{1}{b}-czs\right)\Uw.
\end{align}
Now eliminate one of $U_w(s)$ or $D_w(s)$, and isolate the \textit{kernel} on one side and boundary terms on the other side. We choose to eliminate $D_w$:
\begin{equation}\label{ker}
    \left(1-zsc-\frac{z^2s}{s-cz}\right)U_w(s) = \frac{zs^2}{a(s-cz)}+\frac{zs^2}{s-cz}\left(1-\frac{1}{a}-cz\ols \right)\Dzero + s^w\left(1-\frac{1}{b}-czs\right)\Uw.
\end{equation}
The kernel $K(c;z,s)=\left(1-czs-\frac{z^2s}{s-cz}\right)$ is quadratic in $s$, and has roots
\begin{equation}\label{s(z)}
s_{\pm}(c;z)=\dfrac{1-z^2+c^2z^2\pm\sqrt{(1-z^2+c^2z^2)^2 - 4c^2z^2}}{2cz}
\end{equation}
such that $K(c;z,s_\pm)=0.$ One of these has a power series expansion in $z$ and the other does not:
\begin{align}
    s_-(c;z) &= c z + c z^3 + (c + c^3) z^5 + \bigO(z^7) \\
    s_+(c;z) &= \frac{1}{cz} - \frac{z}{c} - cz^3 - (c+c^3)z^5 + \bigO(z^7).
\end{align}
However, note that since the highest power of $s$ in $U_w(s)$ is $s^w$, the substitution of either root $s=s_\pm$ into $U_w(s)$ leads to a well-defined power series in $z$.

The two roots have a symmetry
\begin{equation}
s_+=\frac{1}{s_-}.
\end{equation}
We thus assign $\hs=s_-$ and substitute $\hs$ and $1/\hs$ into \eqref{ker} to eliminate the left hand side and obtain a system of two equations with two unknowns ($\Dzero$ and $\Uw$). This can now be solved to obtain explicit expressions for $\Dzero$ and $\Uw$:
\begin{align} 
    \Dzero &= \frac{1}{B_w}\left(\hs^{2w}(\hs-cz)(1-b+bcz\hs)+\hs^2(cz\hs-1)((1-b)\hs+bcz)\right) \label{eqn:D0} \\
    \Uw &= \frac{1}{B_w}bc\hs^{w+1}(\hs^2-1)z^2 \label{eqn:Uw}
\end{align}
where
\begin{multline}\label{eqn:Bw_explicit}
    B_w \equiv B_w(a,b,c;z) = \hs^{2w}(\hs-cz)(1-a+acz\hs)(1-b+bcz\hs)\\ +\hs(cz\hs-1)((1-a)\hs+acz)((1-b)\hs+bcz).
\end{multline}

As a side note, eliminating $U_w(s)$ instead of $D_w(s)$, would have given an equation analogous to \eqref{ker}, with $\left(1-\frac{cz}{s}-\frac{z^2}{1-czs}\right)D_w(s)$ on the left hand side. 
This kernel is related to the previous one by $s\to \ols$, reflecting the up/down symmetry of the system, and proceeding with the solution yields the same result as \eqref{eqn:D0}.

\subsection{An alternative solution method}

Instead of using the kernel method to solve $\Dzero$ and $\Uw$, we can also obtain a recursive solution, generating walks in a strip of width $w+1$ by modifying walks in a strip of width $w$. The modification is as follows: take any walk in a strip of width $w$, and for each visit to the top surface, replace it with a (possibly empty) sequence of up-down pairs of steps. (If the last vertex was also a visit to the top surface, one can additionally append a final up step; for bridges, this is mandatory). See \cref{fig:recurrence_illustration}.

\begin{figure}
    \centering
    \resizebox{\textwidth}{!}{
    \begin{tikzpicture}
    \tikzset{bc/.style = {circle, draw, line width=4pt, draw=red, fill=blue, inner sep=3pt}};
    \tikzset{tc/.style = {circle, draw, line width=4pt, draw=green, fill=blue, inner sep=3pt}};
    \tikzset{sp/.style = {draw, line width=3pt, draw=orange, fill=blue, inner sep=4pt}};
    \draw [gray] (0,0) grid (14,4);
    \draw [gray, fill=gray] (0,0) rectangle (14,-0.2);
    \draw [gray, fill=gray] (0,4) rectangle (14,4.2);
    \draw [blue, line width=0.1cm] (0,0) -- (1,1) -- (2,0) -- (3,1) -- (4,2) -- (5,1) -- (6,2) -- (7,3) -- (8,4) -- (9,3) -- (10,4) -- (11,3) -- (12,2) -- (13,3) -- (14,4);
    \node [bc] at (2,0) {};
    \node [tc] at (8,4) {};
    \node [tc] at (10,4) {};
    \node [tc] at (14,4) {};
    \node [sp] at (3,1) {};
    \node [sp] at (6,2) {};
    \node [sp] at (7,3) {};
    \node [sp] at (11,3) {};
    \node [sp] at (13,3) {};
    
    \draw [line width=0.15cm, black, ->] (7,-0.5) -- (7,-1.5);
    
    \begin{scope}[yshift=-7cm, xshift=-5.5cm]
    \draw [gray] (0,0) grid (25,5);
    \draw [gray, fill=gray] (0,0) rectangle (25,-0.2);
    \draw [gray, fill=gray] (0,5) rectangle (25,5.2);
    \draw [blue, line width=0.1cm] (0,0) -- (1,1) -- (2,0) -- (3,1) -- (4,2) -- (5,1) -- (6,2) -- (7,3) -- (8,4);
    \draw [OliveGreen, line width=0.1cm] (8,4) -- (9,5) -- (10,4) -- (11,5) -- (12,4);
    \draw [blue, line width=0.1cm] (12,4) -- (13,3) -- (14,4);
    \draw [OliveGreen, line width=0.1cm] (14,4) -- (15,5) -- (16,4);
    \draw [blue, line width=0.1cm] (16,4) -- (17,3) -- (18,2) -- (19,3) -- (20,4);
    \draw [OliveGreen, line width=0.1cm] (20,4) -- (21,5) -- (22,4) -- (23,5) -- (24,4) -- (25,5);
    \node [bc] at (2,0) {};
    \node [tc] at (9,5) {};
    \node [tc] at (11,5) {};
    \node [tc] at (15,5) {};
    \node [tc] at (21,5) {};
    \node [tc] at (23,5) {};
    \node [tc] at (25,5) {};
    \node [sp] at (3,1) {};
    \node [sp] at (6,2) {};
    \node [sp] at (7,3) {};
    \node [sp] at (8,4) {};
    \node [sp] at (12,4) {};
    \node [sp] at (14,4) {};
    \node [sp] at (16,4) {};
    \node [sp] at (17,3) {};
    \node [sp] at (19,3) {};
    \node [sp] at (20,4) {};
    
    \end{scope}
    
    \end{tikzpicture}
    }
    \caption{Illustrating the recurrence~\eqref{eqn:Uw_recurrence} for $\Uw$. Given a bridge in a strip of width $w$, one can obtain a bridge in the strip of width $w+1$. Each contact with the top wall can be replaced by a `zigzag' path, ie.~a (possibly empty) sequence of up-down pairs of steps. A final up step must be appended at the end. The addition of each non-empty zigzag path generates two new stiffness sites, except the last (which only generates one).}
    \label{fig:recurrence_illustration}
\end{figure}
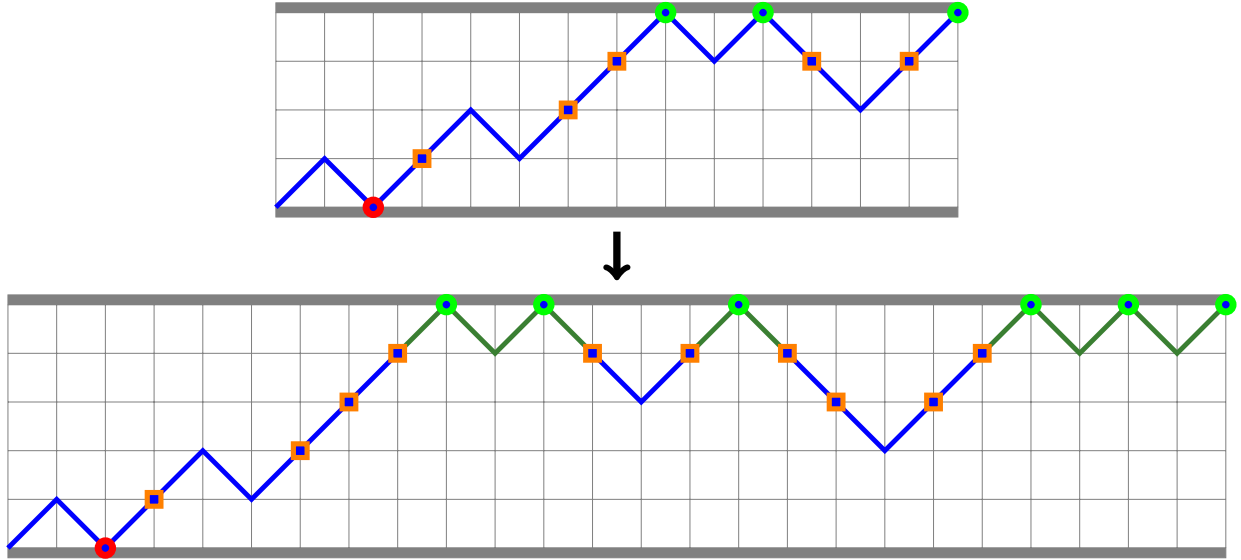

For loops, this leads to the recurrence
\begin{align}
    D_1^{[0]}(a,b,c;z) &= \frac{1}{1-z^2ab} \\
    D_{w+1}^{[0]}(a,b,c;z) &= \Dzero\left(a,1+\frac{z^2bc^2}{1-z^2b},c;z\right)
\end{align}
and for bridges
\begin{align}
    U_1^{[1]}(a,b,c;z) &= \frac{zb}{1-z^2ab} \\
    U_{w+1}^{[w+1]}(a,b,c;z) &= \Uw\left(a,1+\frac{z^2bc^2}{1-z^2b},c;z\right)\times\left(1+\frac{z^2bc^2}{1-z^2b}\right)^{-1}\times\frac{zbc}{1-z^2b} \label{eqn:Uw_recurrence}\\
    &= \frac{zbc}{1-z^2b(1-c^2)}\Uw\left(a,1+\frac{z^2bc^2}{1-z^2b},c;z\right).
\end{align}
A similar but slightly more complicated recurrence can also be found for the total generating function $G_w(a,b,c;z)$.

It follows that all the generating functions we have considered so far are rational. Moreover, by induction one finds that
\begin{align}
    \Dzero(a,b,c;z) &= \frac{P_w(0,b,c;z)}{P_w(a,b,c;z)} \\
    \Uw(a,b,c;z) &= \frac{bc^{w-1}z^w}{P_w(a,b,c;z)} \label{eqn:bridges_recursive_gf}
\end{align}
where the $P_w$ are polynomials satisfying the recurrence
\begin{equation}
    P_{w+1} = (1-z^2(1-c^2))P_w - z^2c^2P_{w-1}
\end{equation}
with $P_1(a,b,c;z) = 1-z^2ab$ and $P_2(a,b,c;z) = 1-z^2(a+b)+z^4ab(1-c^2)$.

Note that by~\eqref{eqn:bridges_recursive_gf} and \cref{lem:fe_all_gfs}, $z_w(a,b,c)$ is the dominant root of $P_w(a,b,c;z)$.

\section{The symmetric case: \texorpdfstring{$a=b$}{a=b}}\label{sec:symmetric}

Visualising the different regions in $(a,b,c)$-phase-space is challenging, so we first dedicate this section to the analysis of the $a=b$ case. Since $\Uw$ has a simpler form~\eqref{eqn:Uw} than $\Dzero$~\eqref{eqn:D0}, we focus on that generating function.

Setting $b\mapsto a$ takes \eqref{eqn:Uw} to
\begin{equation} \label{eqn:Uwsym}
        \Uw = \frac{ac\hs^{w+1}(\hs^2-1)z^2}{\hs^{2w}(\hs-cz)(1-a+acz\hs)^2+\hs(cz\hs-1)((1-a)\hs+acz)^2}.
\end{equation}

\begin{note}
The case $w=2$ is rather trivial and somewhat pathological, so for the remainder of this section we assume that $w\geq3$.
\end{note}

\subsection{The behaviour of \texorpdfstring{$\hs$}{s}}

In order to understand the singularity behaviour of $\Uw$ it will be useful to first briefly discuss $\hs$ (recall~\eqref{s(z)}). It has a power series expansion which is convergent for $|z| < \frac{1}{c+1}$. Since the coefficient of $z^n$ is a polynomial in $c$ with non-negative integer coefficients (this is easily derived from the fact that $K(c;z,\hs) = 0$), it is a strictly increasing function of $z$ on $(0,\frac{1}{c+1}$).

At $z=\frac{1}{c+1}$ we have $\hs = 1$. For $z > \frac{1}{c+1}$, the behaviour depends on $c$. 
\begin{itemize}
    \item If $c=1$ then $\hs$ is complex for all $z>\frac{1}{c+1} = \frac12$.
    \item If $0<c<1$ then $\hs$ is complex on $(\frac{1}{c+1},\frac{1}{1-c})$, equal to $-1$ at $z=\frac{1}{1-c}$, and then real for $z>\frac{1}{1-c}$.
    \item If $c>1$ then $\hs$ is complex on $(\frac{1}{c+1},\frac{1}{c-1})$, equal to $1$ at $z=\frac{1}{c-1}$, and then real for $z>\frac{1}{c-1}$.
\end{itemize}
When $\hs$ is complex, we have
\begin{align}
    \Re(\hs) &= \frac{1-z^2+c^2z^2}{2cz} & \Im(\hs) &= -\frac{\sqrt{4c^2z^2 - (1-z^2+c^2z^2)^2}}{2cz} \\
    & & &= -\sqrt{1-\Re(\hs)^2}
\end{align}
so that $\Re(\hs)^2 + \Im(\hs)^2 = 1$, that is, $\hs$ lies on the unit circle.

\subsection{Zero-force curve}\label{ssec:sym_zfc}

We are interested in understanding the behaviour of the dominant singularity $z_w(a,a,c)$ of $\Uw$ for all $a,c>0$. By \cref{lem:free_energy},~\eqref{eqn:fe(z)} and Pringsheim's theorem~\cite[Thm.~IV.6]{flajolet_analytic_2009}, $z_w(a,a,c)$ is finite, real and positive. Moreover, since $\Uw$ is rational (despite the closed form expression~\eqref{eqn:Uwsym} involving the algebraic function $\hat s$), all singularities are poles of integer order.

By definition the force $\mathcal{F}_w$ is 0 at any point in $a$-$c$ space where $\kappa_w$ (and hence $z_w$) does not depend on $w$. Examining the denominator of~\eqref{eqn:Uwsym}, there are only four ways this can happen: if $z_w$ solves $\hs=0$, $\hs=\pm1$, or
\begin{equation}\label{eqn:symm_zfc_eqn}
    (\hs-cz)(1-a+acz\hs)^2 = \hs(cz\hs-1)((1-a)\hs+acz)^2 = 0.
\end{equation}

There are no relevant solutions to $\hs=0$ or $\hs=-1$. However $\hs=1$ and~\eqref{eqn:symm_zfc_eqn} are both solved when $a=c+1$ and $z=\frac{1}{c+1}$. Note that $z=\frac{1}{c+1}$ is also a root of the numerator of~\eqref{eqn:Uwsym}. Indeed, $\hs^2-1$ has a root of order\footnote{By ``$f(z)$ has a root of order $k$ at $z=z_0$'', we mean $f(z) \sim C(z-z_0)^k$ as $z\to z_0$, for some finite, non-zero constant $C$.} $\frac12$ at $z=\frac{1}{c+1}$. 

In general the root $z=\frac{1}{c+1}$ of the denominator is also of order $\frac12$, leading to a removable singularity. There are two exceptions, however: 
\begin{itemize}
    \item the root is of order $\frac32$ when $a=c+1$, and
    \item the root is of order $\frac32$ when
    \begin{equation}
    a=\frac{(c+1)(w+c-1)}{w-c-1} \qquad\text{and}\qquad c<w-1.
\end{equation}
\end{itemize}
In the latter case $z=\frac{1}{c+1}$ is not the dominant singularity, and we will not discuss it further. In the former case $z=\frac{1}{c+1}$ is the dominant singularity, and we now explain why.

Let $a=c+1$ and suppose that $\Uw$ has a pole at some $z'\in(0,\frac{1}{c+1})$. Since $\hs$ is strictly increasing on this interval, it is invertible. We find its inverse by solving $K(c;z,s) = 0$ in $z$:
\begin{equation}\label{eqn:zpm}
    z_\pm = \frac{c+cs^2\pm\sqrt{4s^2+c^2(1-s^2)^2}}{2s(c^2-1)}.
\end{equation}
The inverse of $\hs$ on $(0,\frac{1}{c+1})$ is then $\hz = z_-$. Substituting $a=c+1$ and $z=\hz$ into~\eqref{eqn:Uwsym} gives
\begin{equation}\label{eqn:Uw_inverse_z}
    \frac{s^w\left(1+s^2+\sqrt{4s^2+c^2(1-s^2)^2}\right)}{c(s^2-1)(s^{2w}-1)}.
\end{equation}
The image of $\hs$ on $(0,\frac{1}{c+1})$ is $(0,1)$, but there are no poles of~\eqref{eqn:Uw_inverse_z} in this interval -- a contradiction. 

It follows that $z_w(c+1,c+1,c) = \frac{1}{c+1}$, or equivalently $\kappa_w(c+1,c+1,c) = \log(c+1)$. Then the force $\mathcal{F}_w(c+1,c+1,c) = 0$.
So $a=c+1$ is the zero-force curve for the symmetric model.

Observe that, as $c$ increases, the zero-force curve increases and thus so too does the region where the force is positive. This intuitively makes sense -- for large $c$, walks tend to prefer to have many long straight segments, and will hence more strongly prefer a wider strip. It follows that the value of $a$ required to induce a negative force must also increase.

\subsection{Off the zero-force curve}\label{ssec:sym_off_zfc}

Away from the zero-force curve $a=c+1$ we are in general unable to exactly compute the dominant singularity and the force. The exception is another curve which, like the zero-force curve, corresponds to a root of the numerator of~\eqref{eqn:Uwsym}. This is the root $z=\frac{1}{c-1}$, which is a simple pole and the dominant singularity if
\begin{equation}
    a=a^*_w=\frac{(c-1)(c-w+1)}{c+w-1} \qquad\text{and}\qquad c>w-1.
\end{equation}
See \cref{fig:symm_astar_curves}. Here we have $\kappa_w(a^*_w,a^*_w,c) = \log(c-1)$, and hence $\mathcal{F}_w(a^*_w,a^*_w,c) = 0$. However, because the curve $a^*_w$ depends on $w$ we do not consider this to be a true ``zero-force curve''.

\begin{figure}
    \centering
    \includegraphics[width=0.6\textwidth]{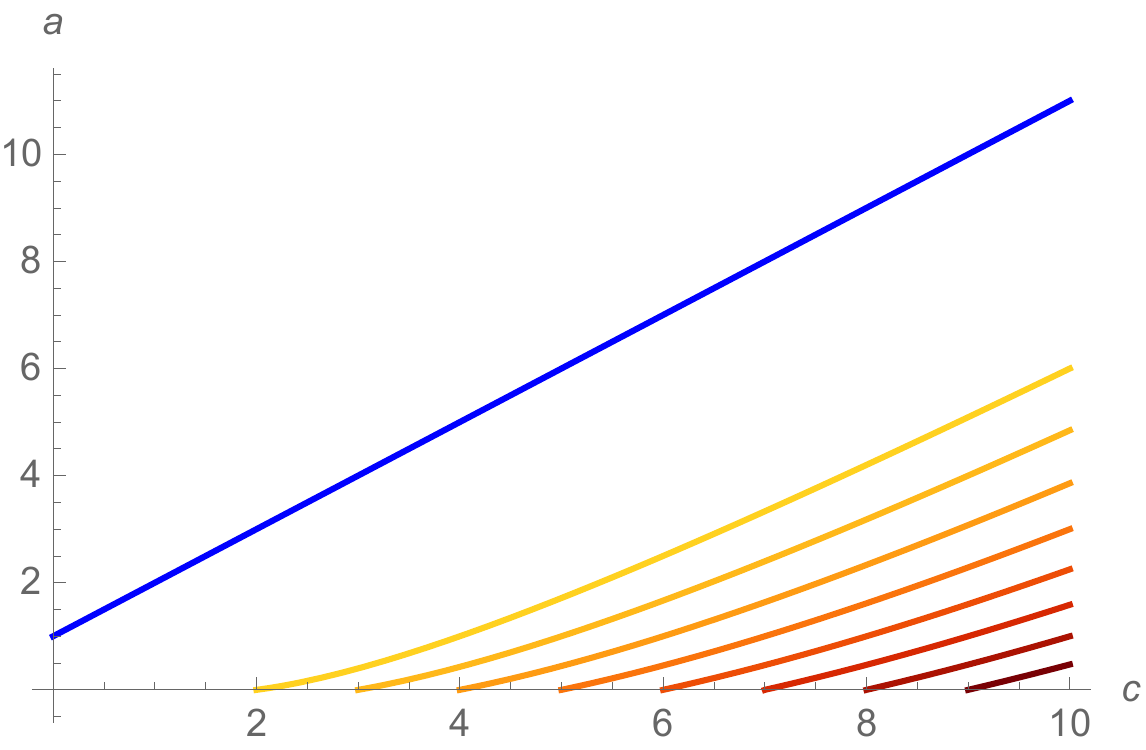}
    \caption{The zero-force curve for the symmetric model (blue), together with the curves $a^*_w$ for $w=3,\dots,10$ (darker colours correspond to larger $w$).}
    \label{fig:symm_astar_curves}
\end{figure}

For other $a,c>0$ we can compute asymptotic expressions for $z_w$ and $\mathcal{F}_w$. We begin by observing that, for fixed $c$, as $a$ decreases the dominant singularity $z_w(a,a,c)$ increases. 
\begin{itemize}
    \item For $a>c+1$ (ie.~the region above the zero-force curve), $z_w(a,a,c) < \frac{1}{c+1}$ and so $\hs$ is real.
    \item Below the zero-force curve, we consider separately the cases $c\leq w-1$ and $c>w-1$.
    \begin{itemize}
        \item[$\circ$] If $c\leq w-1$ then for $a<c+1$ we have $\frac{1}{c+1} < z_w(a,a,c) < \left|\frac{1}{c-1}\right|$, so that $\hs$ is complex and on the unit circle (this is true even for $c=1$).
        \item[$\circ$] If $c>w-1$ then for $a^*_w < a < c+1$ we have $\frac{1}{c+1} < z_w(a,a,c) < \frac{1}{c-1}$, so that $\hs$ is complex and on the unit circle. Then for $0 < a < a^*_w$ we have $z_w(a,a,c) > \frac{1}{c-1}$, so that $\hs$ is again real.
    \end{itemize}
\end{itemize} 

Another way of saying this is that for fixed $c \leq w-1$ there are two distinct regions for $a$: above the zero-force curve (where $\hs$ is real) and below (where $\hs$ is complex). Meanwhile for $c>w-1$ there are three regions for $a$: above the zero-force curve (where $\hs$ is real), between the zero-force curve and $a_w^*$ (where $\hs$ is complex), and below $a_w^*$ (where $\hs$ is again real). However, since
\begin{equation}
    \frac{(c-1)(c-w+1)}{c+w-1} \to 1-c \qquad \text{as}\qquad w\to\infty,
\end{equation}
for any fixed $c$ and $w$ sufficiently large (namely $w \geq c+1$), only the upper two regions exist. Since we will only be computing asymptotic expressions (in $w$), we can assume $w \geq c+1$.

Next, recall that $\hz = z_-$ from~\eqref{eqn:zpm} was the inverse of $\hs$ for $z\in(0,\frac{1}{c+1})$. This is helpful above the zero-force curve, but below things are a little more complicated. To find the boundary of the region where $z_-$ is the inverse of $\hs$, we need to find where the derivative $\frac{\partial}{\partial z}\hs$ is 0. This is solved by $z=\frac{1}{\sqrt{c^2-1}}$. So for given $c>0$, let $a^\dagger_w>0$ be the value of $a$ satisfying $z_w(a^\dagger_w,a^\dagger_w,c) = \frac{1}{\sqrt{c^2-1}}$ (if it exists).

We have been unable to find a simple expression for the curve $a^\dagger_w$ in the $a$-$c$ plane where $z_w(a,a,c)=\frac{1}{\sqrt{c^2-1}}$, however computation readily shows it to lie strictly between the zero-force curve $c+1$ and $a^*_w$. That is, it lies in the complex-$\hs$ region. See \cref{fig:symm_adagger_curves}.

\begin{figure}
    \centering
    \includegraphics[width=0.6\textwidth]{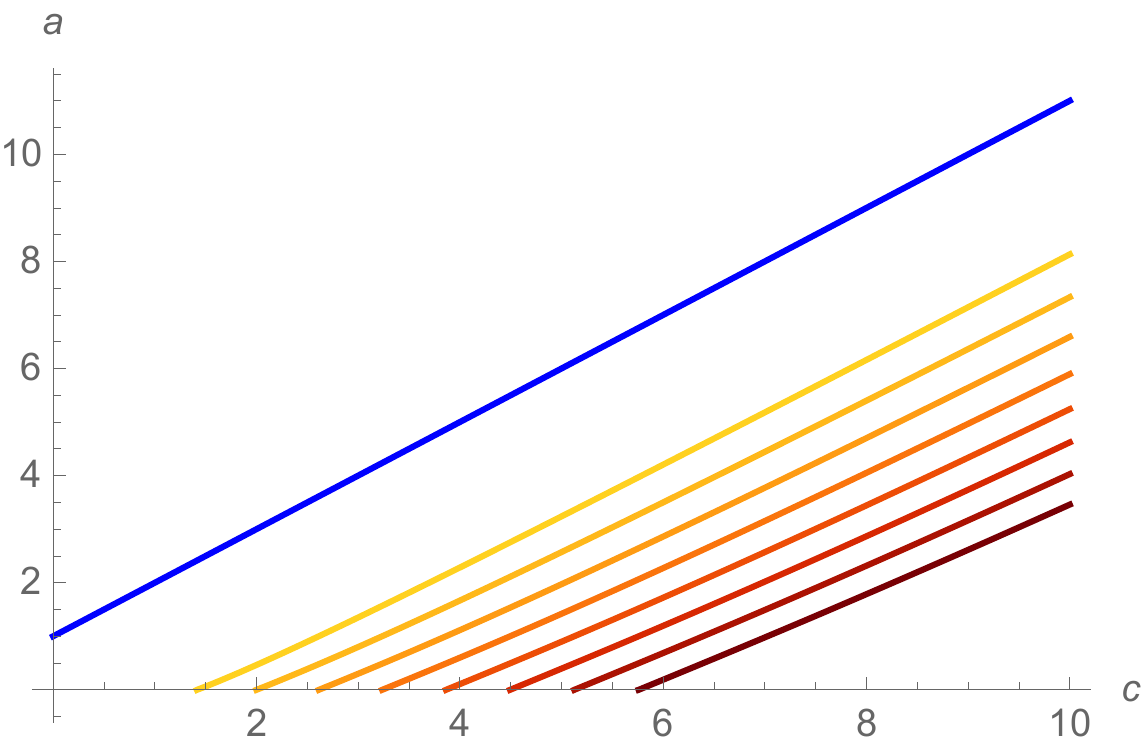}
    \caption{The zero-force curve for the symmetric model (blue), together with the curves $a^\dagger_w$ for $w=3,\dots,10$ (darker colours correspond to larger $w$).}
    \label{fig:symm_adagger_curves}
\end{figure}

Above the curve $a^\dagger_w$ the inverse of $\hs$ is $z_-$, while below $a^\dagger_w$ it is $z_+$. However, for given $c$ the position of $a^\dagger_w$ decreases with $w$ until it drops below 0 (this follows from \cref{thm:strip_to_halfplane}), so that for sufficiently large $w$, the only inverse of $\hs$ for all $a$ is $\hz=z_-$.

\subsubsection{Below the zero-force curve}\label{sssec:below_zfc}

Since we are computing asymptotic approximations, for fixed $c$ we may assume that $w$ is large enough so that $\hs$ is complex and has inverse $\hz$. Then as per previous work~\cite{wong_enumeration_2015} we obtain a good approximation in this region by guessing that $\hs$ is a perturbation of a $2w$-th root of unity
\begin{equation}\label{s_expansion_complex}
    \hs = \exp\left[ \frac{\pi i}{w}\left( c_0+\frac{c_1}{w}+\frac{c_2}{w^2}+\dots \right) \right].
\end{equation}

Take $z=\hz$ in the generating function \eqref{eqn:Uwsym}, and simplify to obtain the denominator. Then, setting the denominator equal to 0 and substituting in \eqref{s_expansion_complex} to solve for the coefficients yields
\begin{equation}
    \hs = \exp\left[\frac{\pi i}{w} \left(-1+\frac{(c+1)(a+c-1)}{(c-a+1)w}-\frac{(c+1)^2(a+c-1)^2}{(c-a+1)^2w^2}+\bigO\left(\frac{1}{w^3}\right) \right) \right].
\end{equation}

This corresponds to a dominant singularity
\begin{equation}
    z_w=\frac{1}{1+c}+\frac{\pi^2c}{2(c+1)w^2}-\frac{\pi^2c(a+c-1)}{(c-a+1)w^3}+\bigO\left(\frac{1}{w^4}\right).
\end{equation}
Using \eqref{eqn:fe(z)}, we find the free energy
\begin{equation}
    \kappa_w=\log(c+1)-\frac{\pi^2c}{2w^2}+\frac{\pi^2c(c+1)(a+c-1)}{(c-a+1)w^3}+\bigO\left(\frac{1}{w^4}\right).
\end{equation}
and the force exerted
\begin{equation}\label{eqn:symm_smalla_force}
    \mathcal{F}_w = \frac{\pi^2c}{w^3}-\frac{3\pi^2c(c+1)(a+c-1)}{(c-a+1)w^4}+\bigO\left(\frac{1}{w^5}\right).
\end{equation}

This is positive and decays as a power law in $w$, which corresponds to repulsive long-range force. 

By \cref{lem:endpoint_height}, the exponent $\nu$ we might expect to fit into~\eqref{eqn:daoud_degennes} is $\frac12$, and indeed the leading term in $\mathcal{F}_w$ matches this exactly.

\subsubsection{Above the zero-force curve}\label{sssec:symm_large_a}

We now turn to the case $a>c+1$. For a singularity, the denominator of \eqref{eqn:Uwsym} must vanish,
\begin{equation}\label{eqn:s_denom}
    \hs^{2w}(\hs-cz)(1-a+acz\hs)^2+\hs(cz\hs-1)((1-a)\hs+acz)^2=0.
\end{equation}
The second term is cancelled by $z=\frac{\sqrt{a-1}}{\sqrt{a(a+c^2-1)}}$, which corresponds to $\hs=\frac{\sqrt{a}c}{\sqrt{(a-1)(a+c^2-1)}}=\frac{acz}{a-1}$. Further, $\left|\hs\right|<1$ on this region, so $\hs^{2w}\to 0$ for $w\to\infty$. Hence this value of $\hs$ solves \eqref{eqn:s_denom} in the limit of large $w$. Writing $\Lambda =\frac{\sqrt{a}c}{\sqrt{(a-1)(a+c^2-1)}}$ and expanding about this value, the next term is exponential in $w$, and must have rate of decay equal to $\Lambda$. Substituting and solving for coefficients, one finds
\begin{equation}\label{s_expansion_real}
    \hs=\Lambda \left[1-\frac{(a-c-1)(a+c-1)(c^2+a^2-1)}{2ac(a-1)(a+c^2-1)}\Lambda^w  + \bigO(\Lambda^{2w}) \right]
\end{equation}

Mapping back to $z$, we find
\begin{equation}
    z_w = \frac{\sqrt{a-1}}{\sqrt{a(a+c^2-1)}}\left[1-\frac{\left(a^2-2 a-c^2+1\right)^2}{2ac (a-1) \left(a+c^2-1\right)}\Lambda^w + \bigO(\Lambda^{2w})\right]
\end{equation}
and hence
\begin{equation}
    \kappa_w = -\frac12\log\left(\frac{a-1}{a(a+c^2-1)}\right) + \frac{\left(a^2-2 a-c^2+1\right)^2}{2ac (a-1)\left(a+c^2-1\right)}\Lambda^w + \bigO(\Lambda^{2w})
\end{equation}
and
\begin{equation}
    \mathcal{F}_w = \frac{\left(a^2-2 a-c^2+1\right)^2\log\Lambda}{2ac (a-1)\left(a+c^2-1\right)}\Lambda^w + \bigO(\Lambda^{2w}).
\end{equation}
This is negative and decays exponentially, corresponding to a short-range attractive force.

Note that for all $a,c>0$, we have $z_w(a,c) \to z^+(a,c)$ as per~\eqref{eqn:halfplane_domsing}.

\section{The asymmetric case}\label{sec:asymmetric}

\subsection{Zero-force surface}

We now turn to general $(a,b,c)$-space. The zero-force `curve' is now really a zero-force `surface'.

For the zero-force surface we follow the same argument as the symmetric case, observing that for the force $\mathcal{F}_w$ to be 0, $z_w$ must solve $\hs=1$ or 
\begin{equation}\label{eqn:asymm_zfc_eqn}
    (\hs-cz)(1-a+acz\hs)(1-b+bcz\hs) = \hs(cz\hs-1)((1-a)\hs+acz)((1-b)\hs+bcz) = 0.
\end{equation}
Both $\hs=1$ and~\eqref{eqn:asymm_zfc_eqn} are solved when $z=\frac{1}{c+1}$ and $a=c+1$ or $b=c+1$, but this is a singularity of $\Uw$ only if we also have $a=b$. However there is another solution to~\eqref{eqn:asymm_zfc_eqn}: when $a,b>1$,
\begin{equation}\label{eqn:asymm_zfc_soln}
    ab-a-b-c^2+1=0 \qquad\text{and}\qquad z = z^* =\frac{\sqrt{a-1}}{\sqrt{a(a+c^2-1)}}.
\end{equation}
(By symmetry one can replace $a$ with $b$ in the equation for $z^*$.) Note that this reduces to known curves ($a=c+1$ and $ab-a-b=0$~\cite{brak_directed_2005}) in the $a=b$ and $c=1$ cases respectively.

Certainly~\eqref{eqn:asymm_zfc_soln} describes a singularity of $\Uw$; it remains to be shown that there is no singularity smaller than $z^*$ when $ab-a-b-c^2+1=0$. This can also be done in the same way as the symmetric case. First note that $z^* < \frac{1}{c+1}$ for $a>1$ and $c>0$, so that $\hs$ is invertible for $0<z\leq z^*$ with inverse $\hz$. Setting $b=\frac{a+c^2-1}{a-1}$ and $z=\hz$ in $\Uw$, the denominator factorises as $(s^{2w}-1)f(s;a,c)$, where
\begin{align}
    f(s;a,c) &= \alpha\sqrt{4s^2 + c^2\left(1-s^2\right)^2} + \beta \\
    \alpha &= \begin{multlined}[t] 2\left(a^2c^2s^4 + a^2c^2s^2 + a^2c^2 + a^2s^2 + ac^4s^4 + ac^4 - ac^2s^4 + 2ac^2s^2\right. \\ \left.- ac^2 - 2as^2 + c^4s^2 - 2c^2s^2 + s^2\right) 
    \end{multlined} \\
    \beta &= \begin{multlined}[t] - 2c\left(s^2 + 1\right)\left(a^2c^2s^4 - a^2c^2s^2 + a^2c^2 + 3a^2s^2 + ac^4s^4 - 2ac^4s^2 + ac^4 - ac^2s^4\right. \\ \left.+ 6ac^2s^2 - ac^2 - 4as^2 + c^4s^2 - 2c^2s^2 + s^2\right).
    \end{multlined}
\end{align}

The only root of $f$ in $(0,1)$ is
\begin{equation}
    s = \begin{cases} \frac{\sqrt{a}c}{\sqrt{(a-1)(a+c^2-1)}} & \text{if } a\geq c+1 \\ 
    \frac{\sqrt{(a-1)(a+c^2-1)}}{\sqrt{a}c} & \text{if } a<c+1 \end{cases}
\end{equation}
which, upon substitution back into $\hz$, exactly corresponds to $z^*$. 

It follows that $ab-a-b-c^2+1=0$ is the zero-force surface for the full asymmetric model. Along this surface $z_w=z^*$, and thus
\begin{equation}
    \kappa_w = \frac12\left(\log a + \log(a+c^2-1) - \log(a-1)\right)
\end{equation}
and $\mathcal{F}_w = 0$. See \cref{fig:asymm_zfs}.

\begin{figure}
    \centering
    \begin{subfigure}{0.6\textwidth}
    \includegraphics[width=\textwidth]{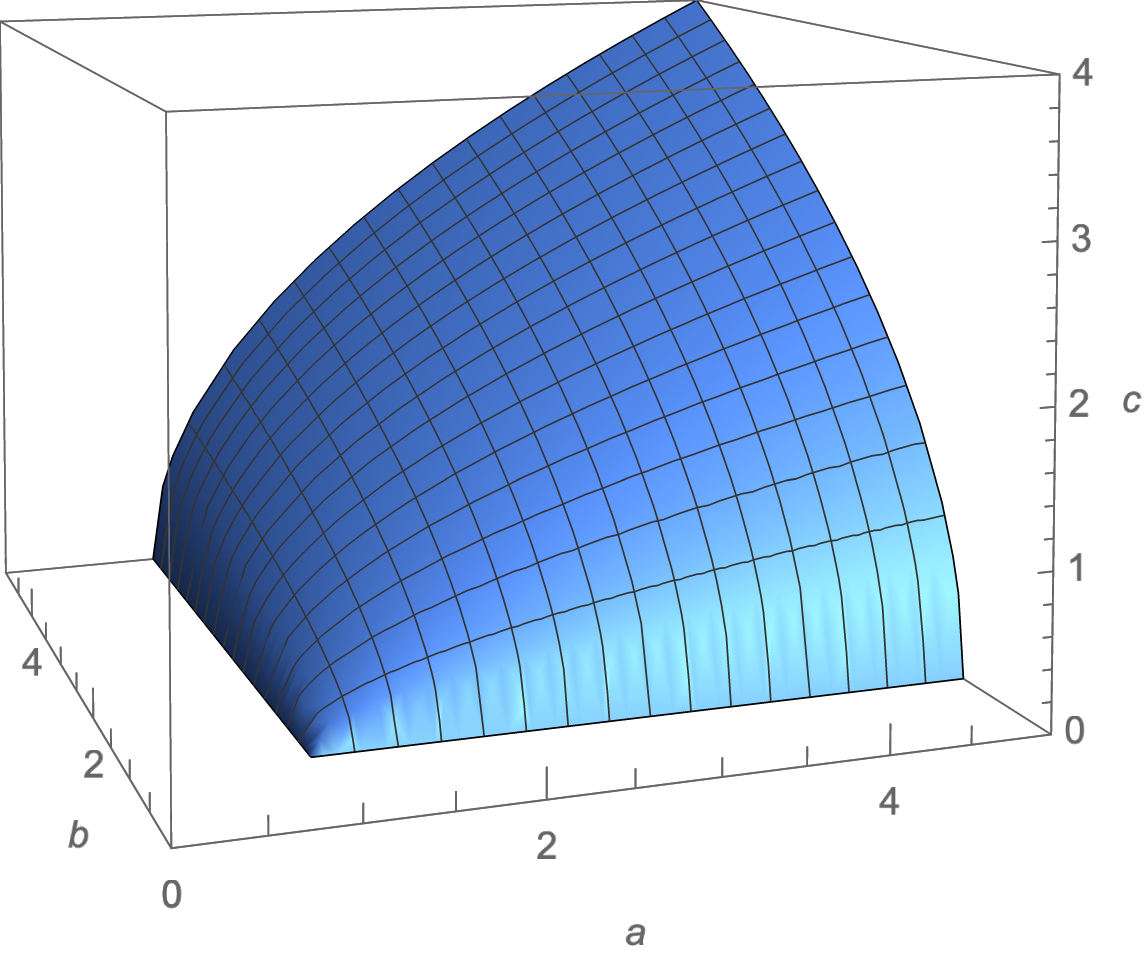}
    \end{subfigure}
    \caption{The zero-force surface in $(a,b,c)$-space.}
    \label{fig:asymm_zfs}
\end{figure}

As we saw in the symmetric case, as $c$ increases the values of $a$ and $b$ required to induce a negative force must also increase.

\subsection{Off the zero-force surface}\label{ssec:asymm_off_zfc}

Away from the zero-force surface, the picture in the asymmetric case is unsurprisingly more complicated than when $a=b$. In the symmetric case, the curve $a=c+1$ was not only the zero-force curve (separating the positive and negative force regions), it also separated the regions where the force was long-range and short-range. In full $(a,b,c)$ space, however, the long-range repulsive and short-range attractive regions only touch along the line $a=b=c+1$, and there are additionally two other regions where the force is short-ranged but repulsive.

As with the symmetric case, the region where the force is long-range corresponds to $\hs(z_w)$ lying on the unit circle, and hence $z_w \in [\frac{1}{c+1},\frac{1}{c-1}]$ (or just $z_w\geq \frac{1}{c+1}$ if $c<1$). Upon substitution we find that the $z_w=\frac{1}{c+1}$ if
\begin{equation}
    b = b^\ddagger_w = \frac{(c+1) (a-aw+(c+1) (c+w-1))}{a (c-w+1)+(c+1) (w-1)} \qquad\text{and}\qquad a<\frac{(c+1) (w-1)}{w-c-1}.
\end{equation}

\begin{figure}
    \centering
    \begin{subfigure}{0.49\textwidth}
    \includegraphics[width=0.8\textwidth]{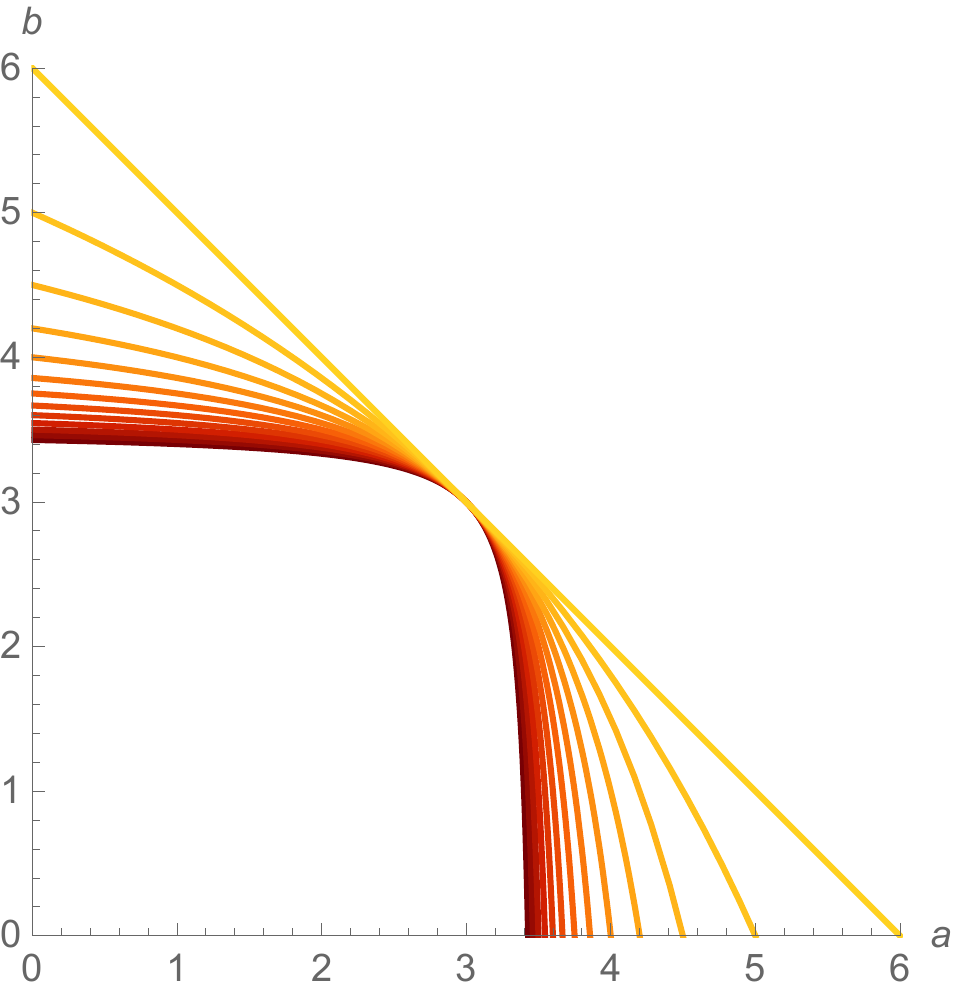}
    \end{subfigure}
    \hfill
    \begin{subfigure}{0.49\textwidth}
    \includegraphics[width=\textwidth]{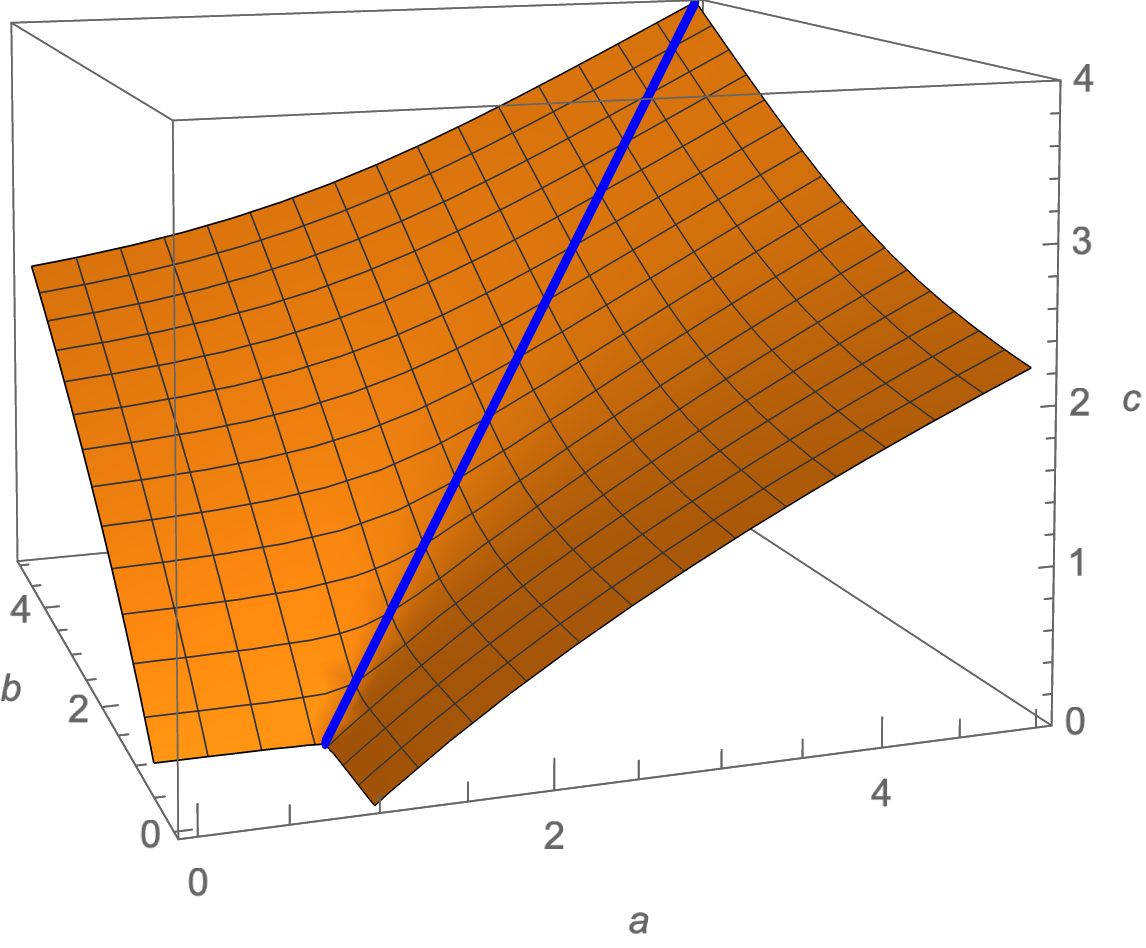}
    \end{subfigure}
    \caption{Two views of $b^\ddagger_w$. \textbf{Left:} Slices of $b^\ddagger_w$ for fixed $c=2$ and $w=3,\dots,15$ (darker colours correspond to larger $w$). \textbf{Right:} The surface $b^\ddagger_5$ in $(a,b,c)$-space. The blue line is $a=b=c+1$, ie.~where the surface touches the zero-force surface.}
    \label{fig:asymm_bddagger}
\end{figure}

See \cref{fig:asymm_bddagger}. As $w\to\infty$ this surface becomes piecewise planar, comprised of $b=c+1$ for $0<a\leq c+1$ and $a=c+1$ for $0<b\leq c+1$. Outside of these planes (that is, if $a>c+1$ or $b>c+1$), the force is short-ranged.

For smaller $a$ or $b$, the analogue of the curve $a^*_w$ is the surface where $z_w = \frac{1}{c-1}$. This occurs if
\begin{equation}
    b=b^*_w = \frac{(c-1) \left(a-a w+c^2-(c-1) w-1\right)}{(c-1) (a+w-1)+a w} \qquad\text{and}\qquad c>w-1.
\end{equation}
Inside of the surface $b^*_w$, we find that $\hs$ is real, while between $b^*_w$ and $b^\ddagger_w$ it is complex and on the unit circle. However, just as we had for the $a=b$ case, for fixed $c$ the surface $b^*_w$ disappears for $w$ sufficiently large, leaving only complex $\hs$ (and thus a short-range force) on the inside of $b^\ddagger_w$ and real $\hs$ (long-range force) on the outside.

Finally, we must consider the equivalent of $a^\dagger_w$, which informed us in the $a=b$ case how to invert $\hs$. Here this is the surface $b^\dagger_w$ defined by $z_w(a,b^\dagger_w,c) = \frac{1}{\sqrt{c^2-1}}$. Analogously to the $a=b$ case, the surface $b^\dagger_w$ lies strictly between $b^*_w$ and $b^\ddagger_w$ (ie.~in the complex $\hs$ region); moreover, it also vanishes for fixed $c$ and $w$ sufficiently large. 
So, for fixed $c$ and $w$ sufficiently large, $\hz=z_-$ is the only inverse of $\hs$ for all $a,b>0$.

\subsubsection{\texorpdfstring{$a,b<c+1$}{a,b<c+1}}\label{sssec:asymm_small_ab}
There are now more cases to consider than before. We begin with the case that both $a$ and $b$ are smaller than $c+1$. Using the same approach as \cref{sssec:below_zfc}, we find
\begin{multline}
    \hs = \exp\left[\frac{\pi i}{w}\left(-1-\frac{(c+1) \left(a b-a-b-c^2+1\right)}{(c-a+1) (c-b+1)w}\right.\right. \\ \left.\left.-\frac{(c+1)^2 \left(a b-a-b-c^2+1\right)^2}{(c-a+1)^2 (c-b+1)^2w^2} + \bigO\left(\frac{1}{w^3}\right)\right)\right].
\end{multline}
It follows that
\begin{equation}
    z_w = \frac{1}{c+1} +\frac{\pi ^2 c}{2(c+1) w^2} + \frac{\pi ^2 c \left(a b-a-b-c^2+1\right)}{ (c-a+1) (c-b+1)w^3 } + \bigO\left(\frac{1}{w^4}\right),
\end{equation}
\begin{equation}
    \kappa_w = \log(c+1) -\frac{\pi ^2 c}{2 w^2} -\frac{\pi ^2 c (c+1) \left(a b-a-b-c^2+1\right)}{ (c-a+1) (c-b+1)w^3} + \bigO\left(\frac{1}{w^4}\right),
\end{equation}
and
\begin{equation}
    \mathcal{F}_w = \frac{\pi^2c}{w^3} + \frac{3 \pi ^2 c (c+1) \left(a b-a-b-c^2+1\right)}{(c-a+1) (c-b+1)w^4} + \bigO\left(\frac{1}{w^5}\right).
\end{equation}
In this region there is thus a long-range repulsive force. Note that the leading term in $\mathcal{F}_w$ is the same as the symmetric case~\eqref{eqn:symm_smalla_force}.

\subsubsection{\texorpdfstring{$a=c+1$}{a=c+1} and \texorpdfstring{$b<c+1$}{b<c+1}}\label{sssec:asymm_small_ab_boundary}

On the boundary between the long-range and short-range regions, a slightly different asymptotic form holds. 
\begin{equation}\label{eqn:asymm_straight_bdy_right_s}
    \hs = \exp\left[\frac{\pi i}{2w}\left(-1 +\frac{(c+1) (c+b-1)}{2 (c-b+1)w} -\frac{(c+1)^2 (c+b-1)^2}{4 (c-b+1)^2w^2} + \bigO\left(\frac{1}{w^3}\right)\right)\right]
\end{equation}
\begin{equation}
    z_w = \frac{1}{c+1} + \frac{\pi^2c}{8(c+1)w^2} - \frac{\pi^2c(c+b-1)}{8(c-b+1)w^3} + \bigO\left(\frac{1}{w^4}\right)
\end{equation}
\begin{equation}
    \kappa_w = \log(c+1) - \frac{\pi^2c}{8w^2} + \frac{\pi^2c(c+1)(c+b-1)}{8(c-b+1)w^3} + \bigO\left(\frac{1}{w^4}\right)
\end{equation}
\begin{equation}\label{eqn:asymm_straight_bdy_right_F}
    \mathcal{F}_w = \frac{\pi^2c}{4w^3} - \frac{3\pi^2c(c+1)(c+b-1)}{8(c-b+1)w^4} + \bigO\left(\frac{1}{w^5}\right)
\end{equation}
The force is thus still long-range and repulsive.

For $b=c+1$ and $a<c+1$, simply switch $a$ and $b$ in~\eqref{eqn:asymm_straight_bdy_right_s}--\eqref{eqn:asymm_straight_bdy_right_F}.

\subsubsection{\texorpdfstring{$a>c+1$}{a>c+1} and \texorpdfstring{$a > b$}{a>b}}\label{sssec:asymm_large_ab}

Finally we turn to the short-range region. We again take the same approach as in the symmetric case. Recall the denominator $B_w$ of $\Uw$ from~\eqref{eqn:Bw_explicit}.
Solving $\hs(cz\hs-1)((1-a)\hs+acz)((1-b)\hs+bcz) = 0$ gives two solutions:
\begin{equation}\label{eqn:asymm_largeab_sroots}
    \hs = \frac{\sqrt{a}c}{\sqrt{(a-1)(a+c^2-1)}} \qquad \text{and}\qquad \hs = \frac{\sqrt{b}c}{\sqrt{(b-1)(b+c^2-1)}}.
\end{equation}

First consider the $a$-dependent solution, and set $\Lambda = \frac{\sqrt{a}c}{\sqrt{(a-1)(a+c^2-1)}}$. Taking the solution to $B_w = 0$ as an expansion about $\hs = \Lambda$, we find that the next term is exponential with rate of decay $\Lambda^2$ (not $\Lambda$, as it was for the symmetric case). Substituting and solving for the coefficients, we find
\begin{equation}\label{eqn:asymm_shortrange_s}
    \hs = \Lambda\left(1 - \frac{((a-1)^2-c^2)(a^2+c^2-1)(ab-a-b-c^2+1)}{2ac^2(a-1)(a-b)(a+c^2-1)}\Lambda^{2w} + \bigO(\Lambda^{4w})\right).
\end{equation}
with corresponding value of $z$
\begin{equation}\label{eqn:asymm_shortrange_z}
    z = \frac{\sqrt{a-1}}{\sqrt{a(a+c^2-1)}}\left(1-\frac{((a-1)^2-c^2)^2(ab-a-b-c^2+1)}{2ac^2(a-1)(a-b)(a+c^2-1)}\Lambda^{2w} + \bigO(\Lambda^{4w})\right).
\end{equation}

By symmetry, had we taken the second solution in~\eqref{eqn:asymm_largeab_sroots}, the corresponding expansions for $\hs$ and $z$ could be found by swapping $a$ and $b$ in~\eqref{eqn:asymm_shortrange_s} and~\eqref{eqn:asymm_shortrange_z}. Now
\begin{equation}
    \frac{\sqrt{a-1}}{\sqrt{a(a+c^2-1)}} < \frac{\sqrt{b-1}}{\sqrt{b(b+c^2-1)}} \iff a>b,
\end{equation}
so when $a>b$ the dominant singularity $z_w$ is given by~\eqref{eqn:asymm_shortrange_z}. Then
\begin{equation}
    \kappa_w = -\frac12\log\left(\frac{a-1}{a(a+c^2-1)}\right) + \frac{((a-1)^2-c^2)^2(ab-a-b-c^2+1)}{2ac^2(a-1)(a-b)(a+c^2-1)}\Lambda^{2w} + \bigO(\Lambda^{4w})
\end{equation}
and
\begin{equation}\label{eqn:asymm_shortrange_F}
    \mathcal{F}_w = \frac{((a-1)^2-c^2)^2(ab-a-b-c^2+1)\log\Lambda}{ac^2(a-1)(a-b)(a+c^2-1)}\Lambda^{2w} + \bigO(\Lambda^{4w})
\end{equation}
The force is thus short-range in this region. Note that $\mathcal{F}_w$ is positive for $b<\frac{a+c^2-1}{a-1}$ (that is, `inside' the zero-force surface) and negative if $b>\frac{a+c^2-1}{a-1}$.

For the reflected region with $b>c+1$ and $b>a$, simply swap $a$ and $b$ in~\eqref{eqn:asymm_shortrange_s}--\eqref{eqn:asymm_shortrange_F}.

Note that for all $a,b,c>0$, we have $z_w(a,b,c) \to \min\{z^+(a,c),z^+(b,c)\}$ as per~\eqref{eqn:halfplane_domsing}, and hence $\kappa_w(a,b,c) \to \max\{\kappa^+(a,c),\kappa^+(b,c)\}$. See \cref{fig:asymm_zfc_regions_slice} for an illustration of the different regions when $c=2$.

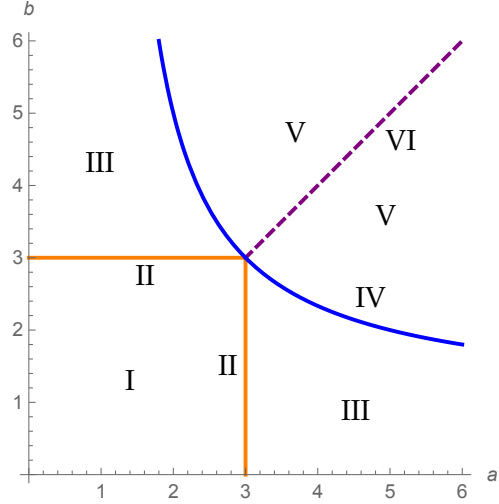
\begin{figure}
    \centering
    \begin{subfigure}{0.49\textwidth}
    \begin{tikzpicture}
    \node at (0,0) [anchor=south west] {\includegraphics[width=0.8\textwidth]{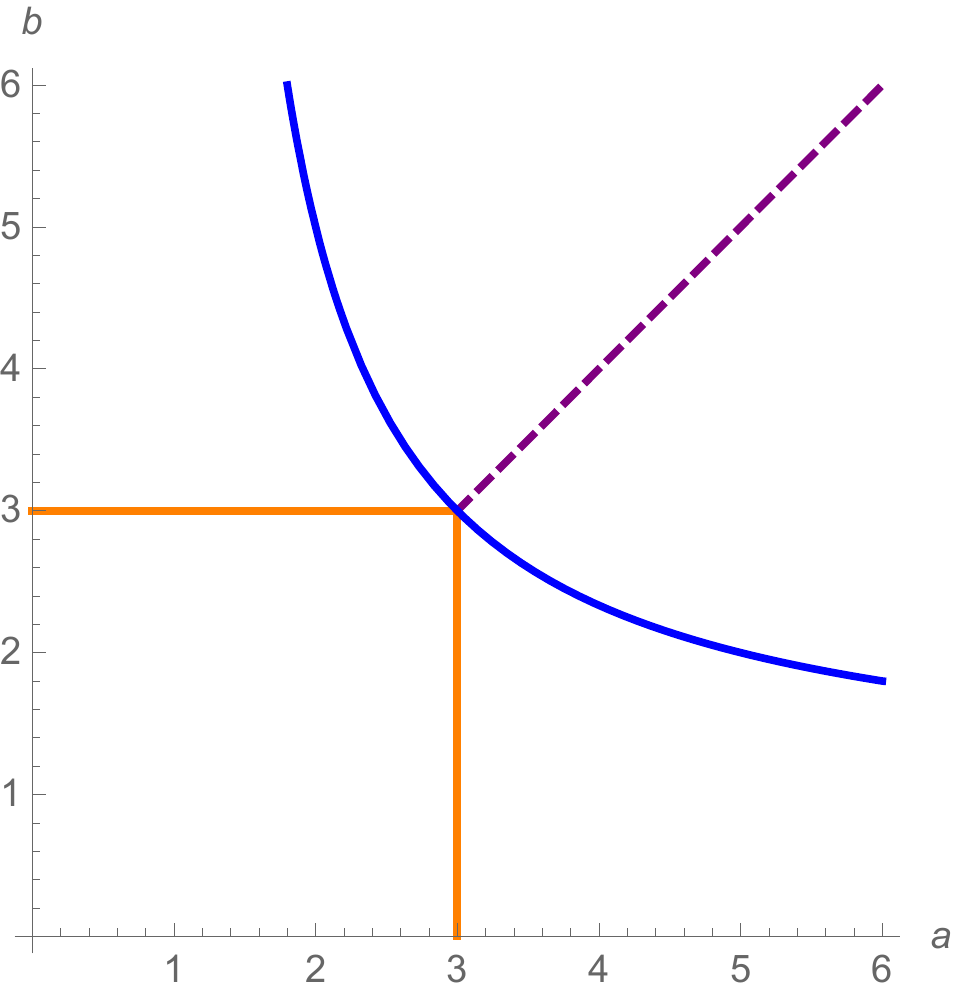}};
    \node at (1.7,1.8) {I};
    \node at (3,2) {II};
    \node at (1.9,3.2) {II};
    \node at (4.7,1.4) {III};
    \node at (1.3,4.7) {III};
    \node at (4.9,2.9) {IV};
    \node at (5.1,4) {V};
    \node at (3.9,5.1) {V};
    \node at (5.3,5) {VI};
    \end{tikzpicture}
    \end{subfigure}
    \caption{The different regions in a plane for fixed $c$ (in this case, $c=2$). The zero-force surface intersects the plane along the blue curve (IV). In region I and along its boundary (II) the force is positive and long-ranged (see \cref{sssec:asymm_small_ab} and \cref{sssec:asymm_small_ab_boundary} respectively). In region III the force is positive and short-ranged, while in region V it is negative and short-ranged (\cref{sssec:asymm_large_ab}). Along the line (VI) it is also negative and short-ranged, but a slightly different asymptotic form holds (\cref{sssec:symm_large_a}).} 
    \label{fig:asymm_zfc_regions_slice}
\end{figure}

\section{Sampling}\label{sec:sampling}

The \emph{Boltzmann distribution} assigns probability
\begin{equation}
    \mathbb{P}(\phi) = \frac{a^{m_a(\phi)}b^{m_b(\phi)}c^{m_c(\phi)}}{Z_{w,n}(a,b,c)}
\end{equation}
to a walk $\phi$. There are multiple ways to sample directly from this distribution, and even more if one is satisfied with only approximating it. One direct method involves computing the dominant eigenvalue and corresponding eigenvector of the transfer matrix~\cite{alm_random_1990}, while \emph{Boltzmann sampling} \cite{duchon_boltzmann_2004} can be used to generate objects of random size (but with correct relative probabilities within a given size).

We have implemented another method, known as the \emph{generating tree} method\footnote{Despite the fact that our underlying graph structure is not a tree.} \cite{NijenhuisWilf}. To sample objects of size $n$, one computes a labelled graph $G$ with $n+1$ levels, along with a weight function $F : V(G) \to \mathbb{R}$. The graph $G$ is essentially a graphical representation of the powers of the transfer matrix -- a node at level $m<n$ with a given label corresponds to a set of walks of length $m$, which can all be extended (by the addition of a step) in the same way, and accrue the same weight with each extension. Each different extension then corresponds to a different `child' at level $m+1$ (but multiple nodes at level $m$ could share the same child at level $m+1$). For a node $v$ with label $\ell$, the function $F(v)$ is then the sum of the total weights of all possible ``completions'' from walks with label $\ell$.

In our case, each node gets label $(h,p)$ where $h$ is an integer between 0 and $w$ (corresponding to the endpoint height of a walk) and $p$ is one of $\{\text{DD},\text{DU},\text{UD},\text{UU}\}$ (corresponding to the directions of the last two steps of a walk). We assign labels $(0,--)$ and $(1,-\text{U})$ to the nodes at level 0 and 1 respectively. See \cref{fig:generating_tree} for an illustration of $G$ when $w=2$ and $n=5$.

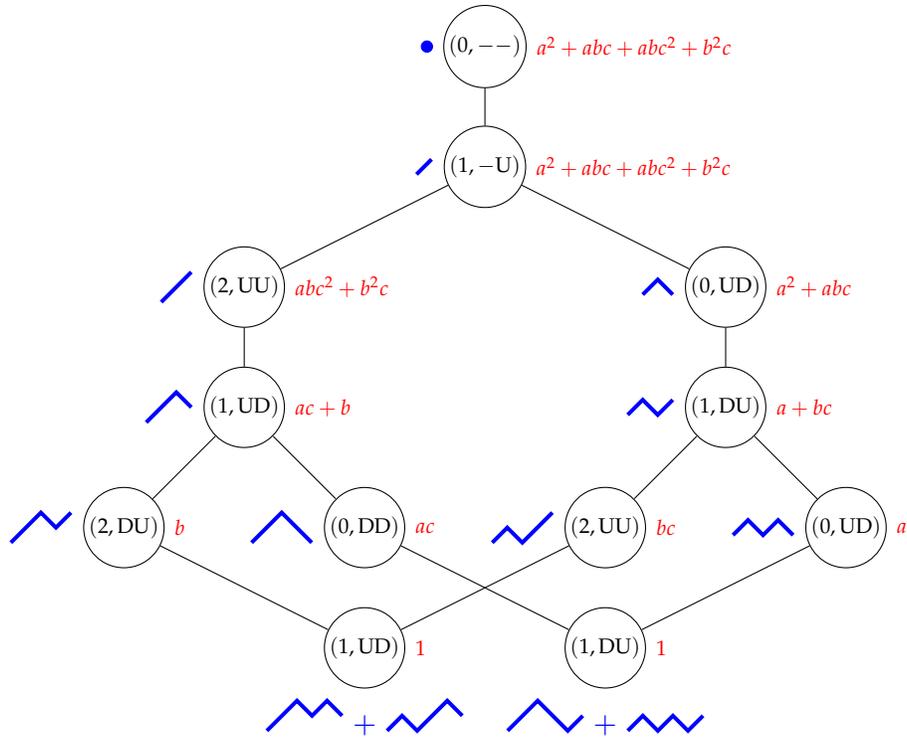
\begin{figure}[ht]
    \centering
    \begin{tikzpicture}[scale=0.8]
    \node (A) [draw, circle, inner sep=1pt, label={[red] right:{\scriptsize $a^2+abc+abc^2+b^2c$}}, label={left:\tikz{\node [circle, draw=blue, fill=blue, inner sep=1.5pt] at (0,0) {};}}] at (0,0) {\scriptsize $(0,--)$};
    \node (B) [draw, circle, inner sep=1pt, label={[red] right:{\scriptsize $a^2+abc+abc^2+b^2c$}}, label={left:\tikz{\draw [line width=1.5pt, blue] (0,0) -- (0.2,0.2);}}] at (0,-2) {\scriptsize $(1,-\text{U})$};
    \node (C) [draw, circle, inner sep=1pt, label={[red] right:{\scriptsize $abc^2+b^2c$}}, label={left:\tikz{\draw [line width=1.5pt, blue] (0,0) -- (0.2,0.2) -- (0.4,0.4);}}] at (-4,-4) {\scriptsize $(2,\text{UU})$};
    \node (D) [draw, circle, inner sep=1pt, label={[red] right:{\scriptsize $a^2+abc$}}, label={left:\tikz{\draw [line width=1.5pt, blue] (0,0) -- (0.2,0.2) -- (0.4,0);}}] at (4,-4) {\scriptsize $(0,\text{UD})$};
    \node (E) [draw, circle, inner sep=1pt, label={[red] right:{\scriptsize $ac+b$}}, label={left:\tikz{\draw [line width=1.5pt, blue] (0,0) -- (0.2,0.2) -- (0.4,0.4) -- (0.6,0.2);}}] at (-4,-6) {\scriptsize $(1,\text{UD})$};
    \node (F) [draw, circle, inner sep=1pt, label={[red] right:{\scriptsize $a+bc$}}, label={left:\tikz{\draw [line width=1.5pt, blue] (0,0) -- (0.2,0.2) -- (0.4,0) -- (0.6,0.2);}}] at (4,-6) {\scriptsize $(1,\text{DU})$};
    \node (G) [draw, circle, inner sep=1pt, label={[red] right:{\scriptsize $b$}}, label={left:\tikz{\draw [line width=1.5pt, blue] (0,0) -- (0.2,0.2) -- (0.4,0.4) -- (0.6,0.2) -- (0.8,0.4);}}] at (-6,-8) {\scriptsize $(2,\text{DU})$};
    \node (H) [draw, circle, inner sep=1pt, label={[red] right:{\scriptsize $ac$}}, label={left:\tikz{\draw [line width=1.5pt, blue] (0,0) -- (0.2,0.2) -- (0.4,0.4) -- (0.6,0.2) -- (0.8,0);}}] at (-2,-8) {\scriptsize $(0,\text{DD})$};
    \node (I) [draw, circle, inner sep=1pt, label={[red] right:{\scriptsize $bc$}}, label={left:\tikz{\draw [line width=1.5pt, blue] (0,0) -- (0.2,0.2) -- (0.4,0) -- (0.6,0.2) -- (0.8,0.4);}}] at (2,-8) {\scriptsize $(2,\text{UU})$};
    \node (J) [draw, circle, inner sep=1pt, label={[red] right:{\scriptsize $a$}}, label={left:\tikz{\draw [line width=1.5pt, blue] (0,0) -- (0.2,0.2) -- (0.4,0) -- (0.6,0.2) -- (0.8,0);}}] at (6,-8) {\scriptsize $(0,\text{UD})$};
    \node (K) [draw, circle, inner sep=1pt, label={[red] right:{\scriptsize $1$}}, label={below:\tikz{\draw [line width=1.5pt, blue] (0,0) -- (0.2,0.2) -- (0.4,0.4) -- (0.6,0.2) -- (0.8,0.4) -- (1,0.2);} {\color{blue} $+$} \tikz{\draw [line width=1.5pt, blue] (0,0) -- (0.2,0.2) -- (0.4,0) -- (0.6,0.2) -- (0.8,0.4) -- (1,0.2);}}] at (-2,-10) {\scriptsize $(1,\text{UD})$};
    \node (L) [draw, circle, inner sep=1pt, label={[red] right:{\scriptsize $1$}}, label={below:\tikz{\draw [line width=1.5pt, blue] (0,0) -- (0.2,0.2) -- (0.4,0.4) -- (0.6,0.2) -- (0.8,0) -- (1,0.2);} {\color{blue} $+$} \tikz{\draw [line width=1.5pt, blue] (0,0) -- (0.2,0.2) -- (0.4,0) -- (0.6,0.2) -- (0.8,0) -- (1,0.2);}}] at (2,-10) {\scriptsize $(1,\text{DU})$};
    \draw (A) -- (B) -- (C) -- (E) -- (G) -- (K);
    \draw (B) -- (D) -- (F) -- (I) -- (K);
    \draw (F) -- (J) -- (L) -- (H) -- (E);
    \end{tikzpicture}
    \caption{The graph $G$ for $w=2$ and $n=5$. For each node $v$ the function $F(v)$ is given in red, while the corresponding set of paths is drawn in blue.}
    \label{fig:generating_tree}
\end{figure}

Once the graph $G$ and set of weights $F$ have been computed, sampling from the Boltzmann distribution is straightforward. Start at the top (level 0), and then at each level choose one of the current node's children with probability proportional to that child's weight $F$. 

In \cref{fig:samples} we illustrate some walks of length 400 in the strip of width 10, for a few different values of $(a,b,c)$.

\begin{figure}[ht]
\centering
\begin{subfigure}{\textwidth}
\includegraphics[width=\textwidth]{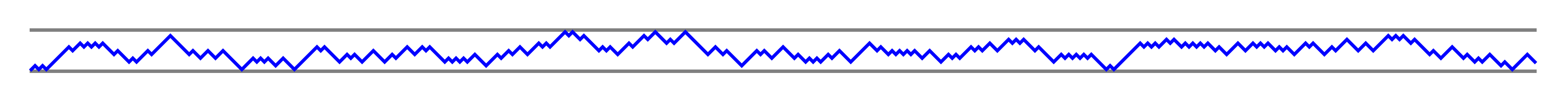}
\caption{$(a,b,c) = (1,1,1)$}
\end{subfigure}
\begin{subfigure}{\textwidth}
\includegraphics[width=\textwidth]{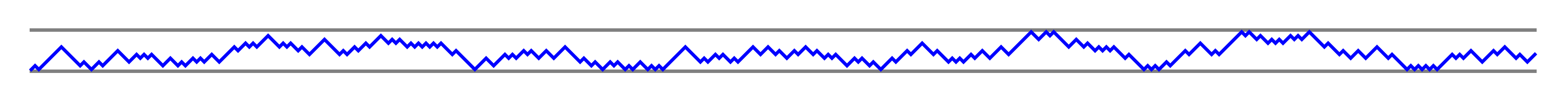}
\caption{$(a,b,c) = (2,2,1)$}
\end{subfigure}
\begin{subfigure}{\textwidth}
\includegraphics[width=\textwidth]{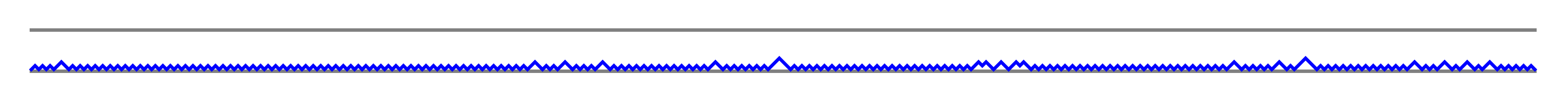}
\caption{$(a,b,c) = (10,10,1)$}
\end{subfigure}
\begin{subfigure}{\textwidth}
\includegraphics[width=\textwidth]{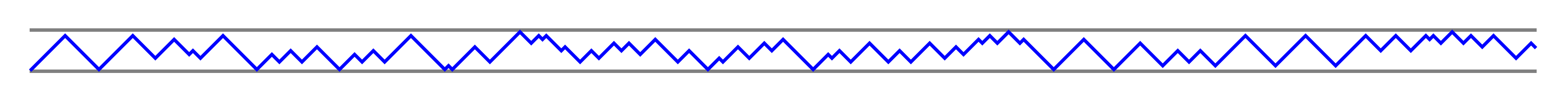}
\caption{$(a,b,c) = (1,1,5)$}
\end{subfigure}
\begin{subfigure}{\textwidth}
\includegraphics[width=\textwidth]{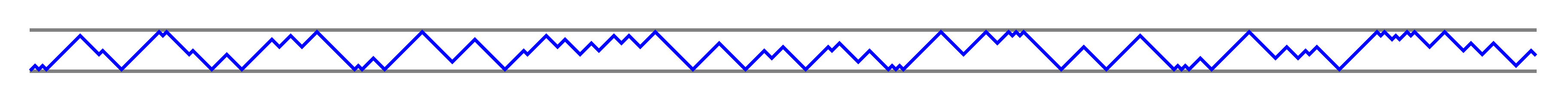}
\caption{$(a,b,c) = (6,6,5)$}
\end{subfigure}
\begin{subfigure}{\textwidth}
\includegraphics[width=\textwidth]{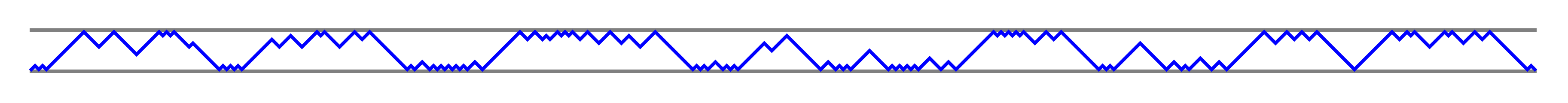}
\caption{$(a,b,c) = (20,20,5)$}
\end{subfigure}
\begin{subfigure}{\textwidth}
\includegraphics[width=\textwidth]{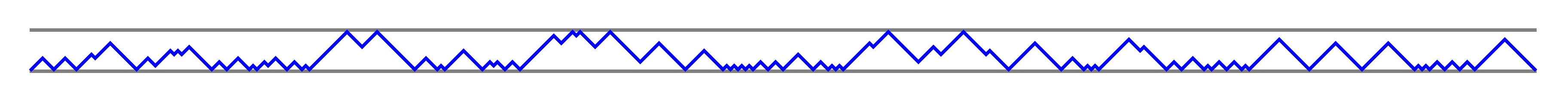}
\caption{$(a,b,c) = (10,\frac{34}{9},5)$}
\end{subfigure}
\caption{Walks of length 400 in the strip of width 10, sampled from the Boltzmann distribution at various values of $(a,b,c)$.}
\label{fig:samples}
\end{figure}

\section{Conclusion}\label{sec:conclusion}

We have defined, solved and analysed a model of semiflexible linear polymers in a strip, interacting with the two walls of the strip. Along the surface $ab-a-b-c^2+1=0$ in $(a,b,c)$-space the polymers exert zero net force on the walls of the strip, while on either side of this surface the polymers work to either push the walls apart or pull them together. As $c$ is increased, the values of $a$ and $b$ required to induce a negative force (that is, to pull the walls together) also increases.

There are a number of possible ways this work can be extended or generalised. The most obvious way is to move from directed walks to SAWs; however, that model is not solvable for general $w$ using current technology (for very small $w$ the transfer matrix can be computed exactly). Monte Carlo methods may yield useful results, however. A more modest extension might involve Motzkin paths (which allow a horizontal step in addition to the diagonal steps used here) or partially directed walks (using steps $(1,0), (0,1)$ and $(0,-1)$).

Instead of (or in addition to) modelling semiflexible polymers, one can model self-interacting polymers by assigning a weight (say, $u$) to each nearest-neighbour pair of occupied sites. The effect of increasing $u$ should be qualitatively similar to increasing $c$ -- for large $u$, polymers will tend to form compact `globules', and this will serve to push the walls apart more strongly. 

\section*{Acknowledgements}

NRB is supported by Australian Research Council grant DE170100186. JL and LL were supported by Vacation Research Scholarships from the Australian Mathematical Sciences Institute. 

\sloppy
\printbibliography
\fussy

\end{document}